\documentclass[12pt]{amsart}
\usepackage{epsf,amscd,amsmath,amssymb,amsfonts}
\usepackage{fullpage}
\newcommand{\be}{\begin{equation}}
\newcommand{\ee}{\end{equation}}
\newcommand{\bea}{\begin{eqnarray}}
\newcommand{\eea}{\end{eqnarray}}
\newcommand{\beas}{\begin{eqnarray*}}
\newcommand{\eeas}{\end{eqnarray*}}

\theoremstyle{plain}
\newtheorem{thm}{Theorem}
\newtheorem{lem}[thm]{Lemma}
\newtheorem{rem}[thm]{Remark}
\newtheorem{cor}[thm]{Corollary}
\newtheorem{prop}[thm]{Proposition}

\theoremstyle{definition}
\newtheorem{defn}[thm]{Definition}
\newtheorem{rmk}[thm]{Remark}

\newtheorem{ex}[thm]{Example}

\numberwithin{thm}{section}
\numberwithin{equation}{section}

\newcommand{\ve}{\varepsilon}
\newcommand{\eq}[2]{\begin{equation}\label{#1}#2 \end{equation}}
\newcommand{\ml}[2]{\begin{multline}\label{#1}#2 \end{multline}}
\newcommand{\ga}[2]{\begin{gather}\label{#1}#2 \end{gather}}

\newcommand{\surj}{\twoheadrightarrow}
\newcommand{\inj}{\hookrightarrow}

\newcommand{\sA}{{\mathcal A}}

\newcommand{\sH}{{\mathcal H}}

\newcommand{\sO}{{\mathcal O}}

\newcommand{\A}{{\mathbb A}}

\newcommand{\C}{{\mathbb C}}
\newcommand{\D}{{\mathbb D}}

\renewcommand{\P}{{\mathbb P}}
\newcommand{\Q}{{\mathbb Q}}
\newcommand{\R}{{\mathbb R}}

\newcommand{\Z}{{\mathbb Z}}

\def\triangleGeom{{\;\raisebox{-3cm}{\epsfysize=8cm\epsfbox{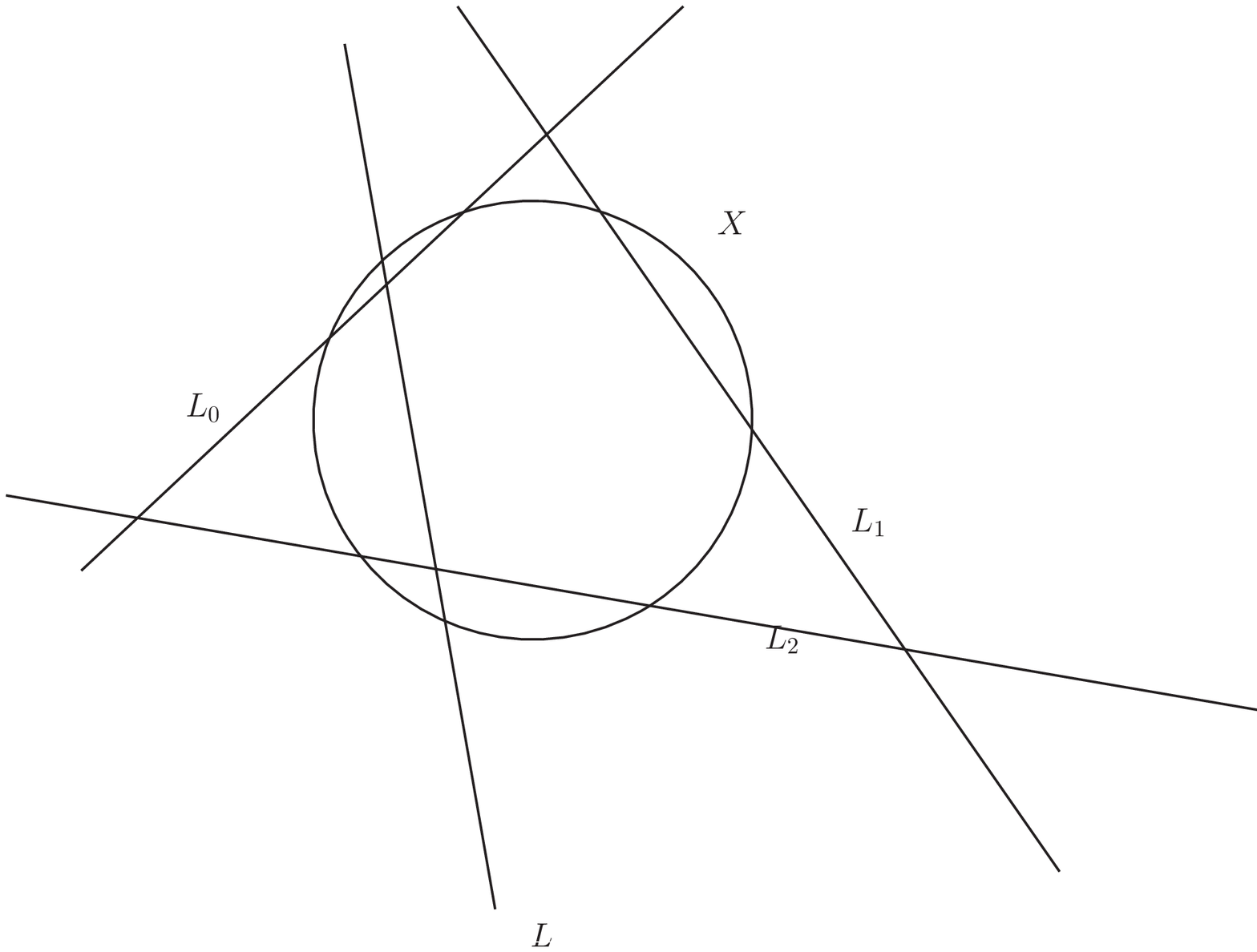}}\;}}

\def\triangle{{\;\raisebox{-5mm}{\epsfysize=1cm\epsfbox{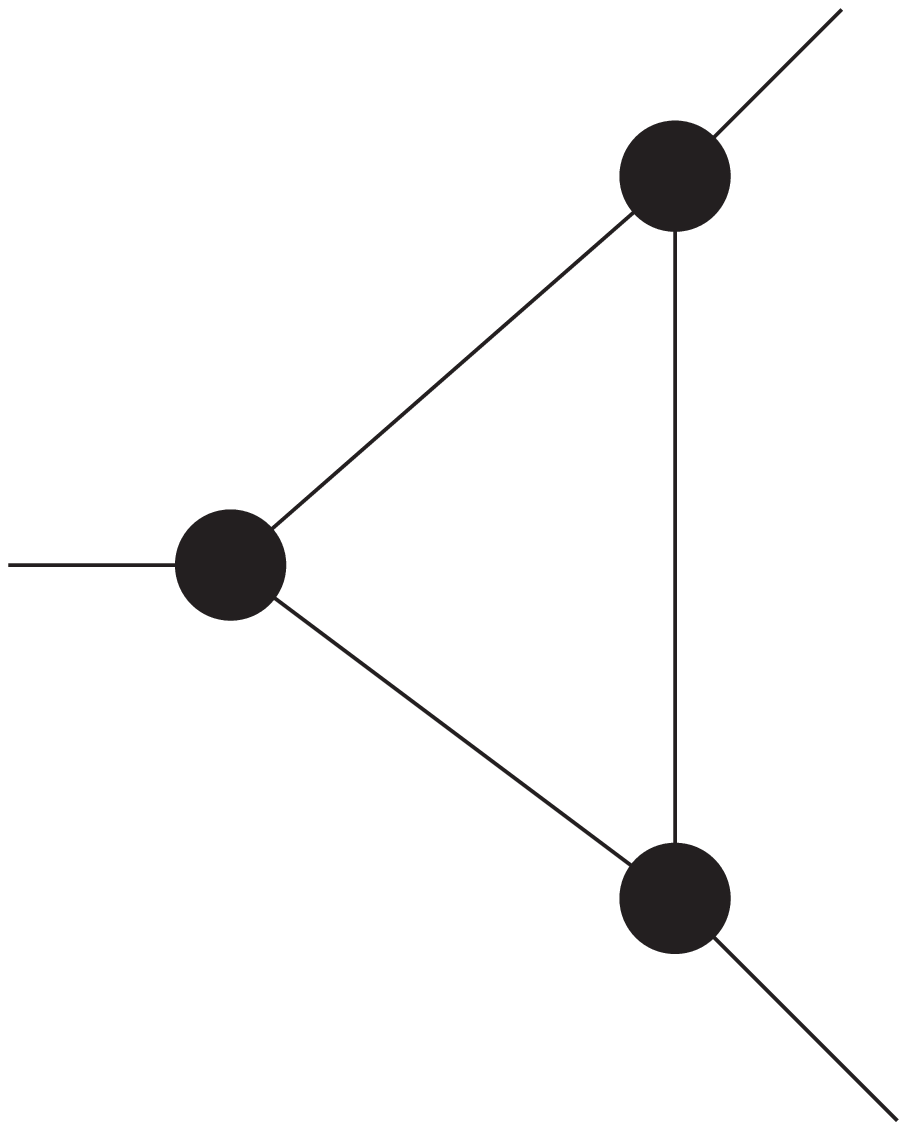}}\;}}
\def\triangleab{{\;\raisebox{-5mm}{\epsfysize=1cm\epsfbox{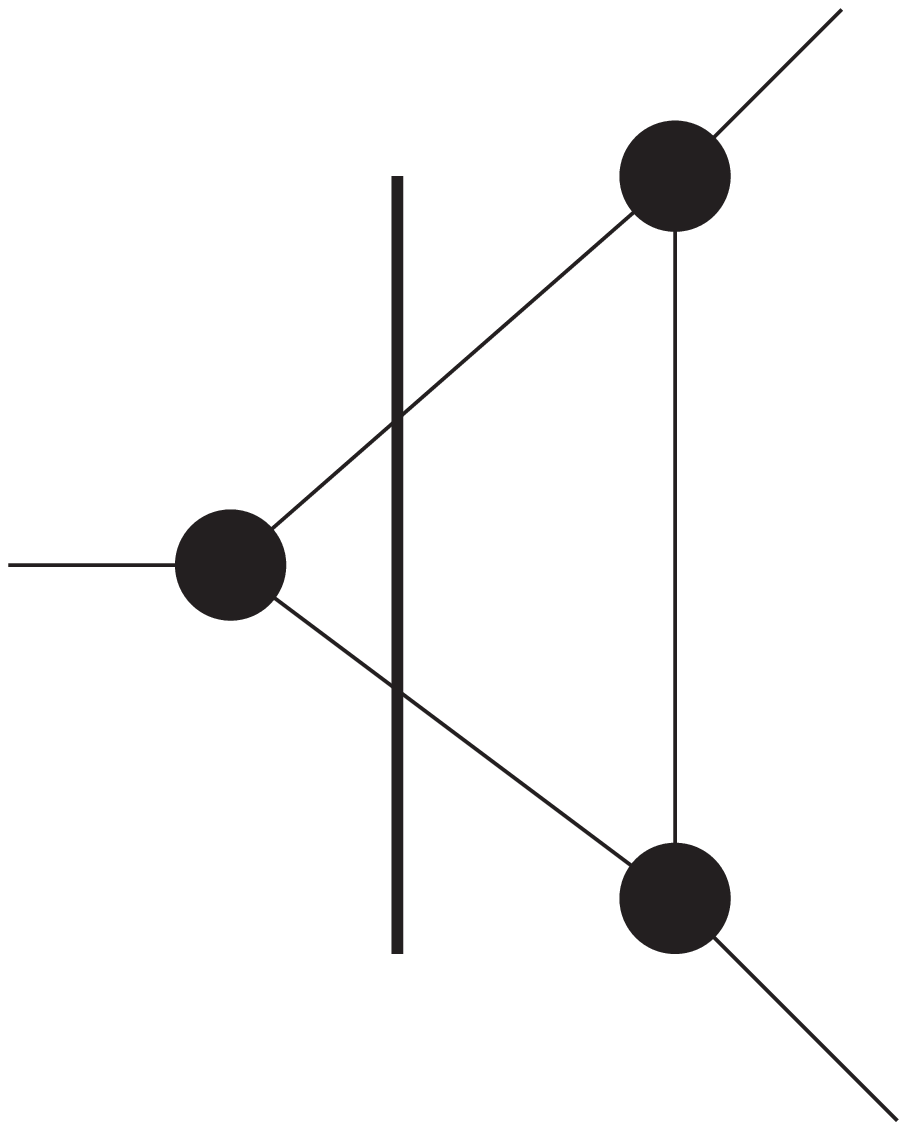}}\;}}
\def\trianglebc{{\;\raisebox{-5mm}{\epsfysize=1cm\epsfbox{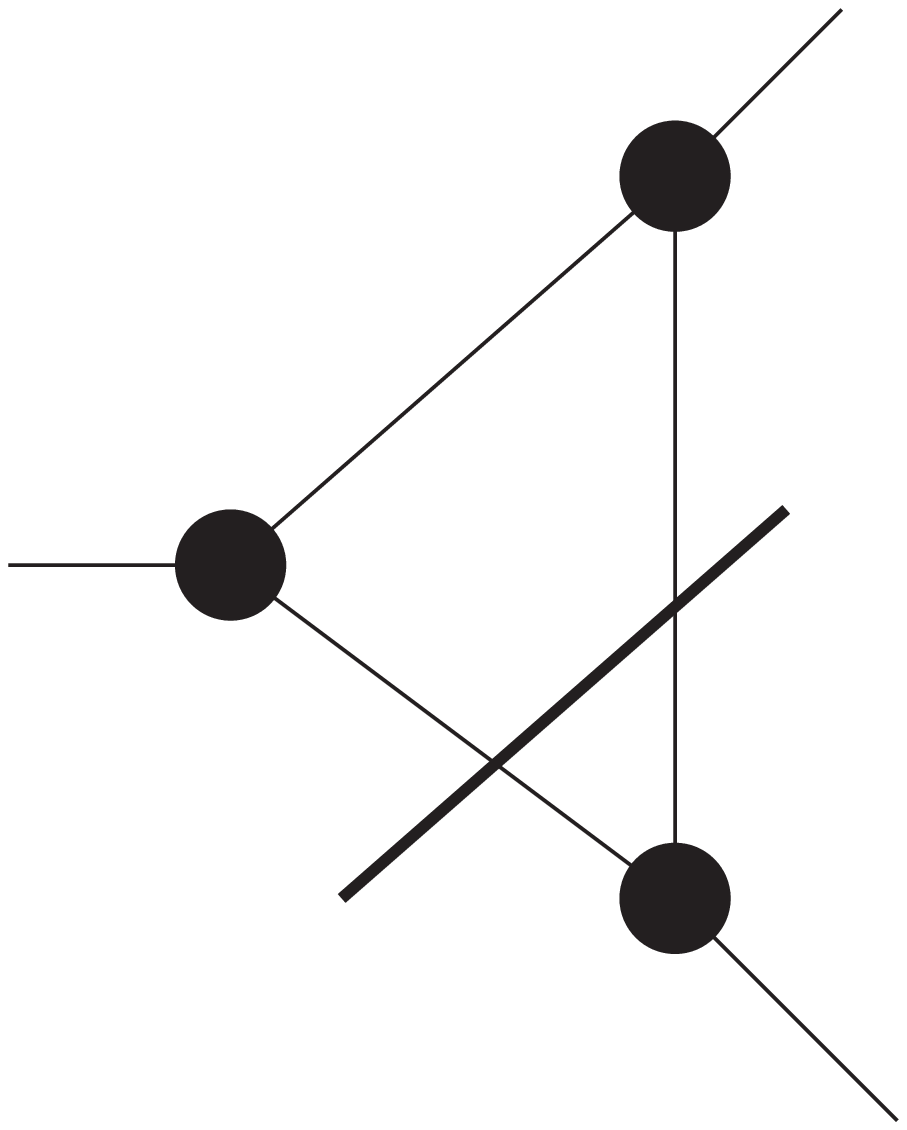}}\;}}
\def\triangleca{{\;\raisebox{-5mm}{\epsfysize=1cm\epsfbox{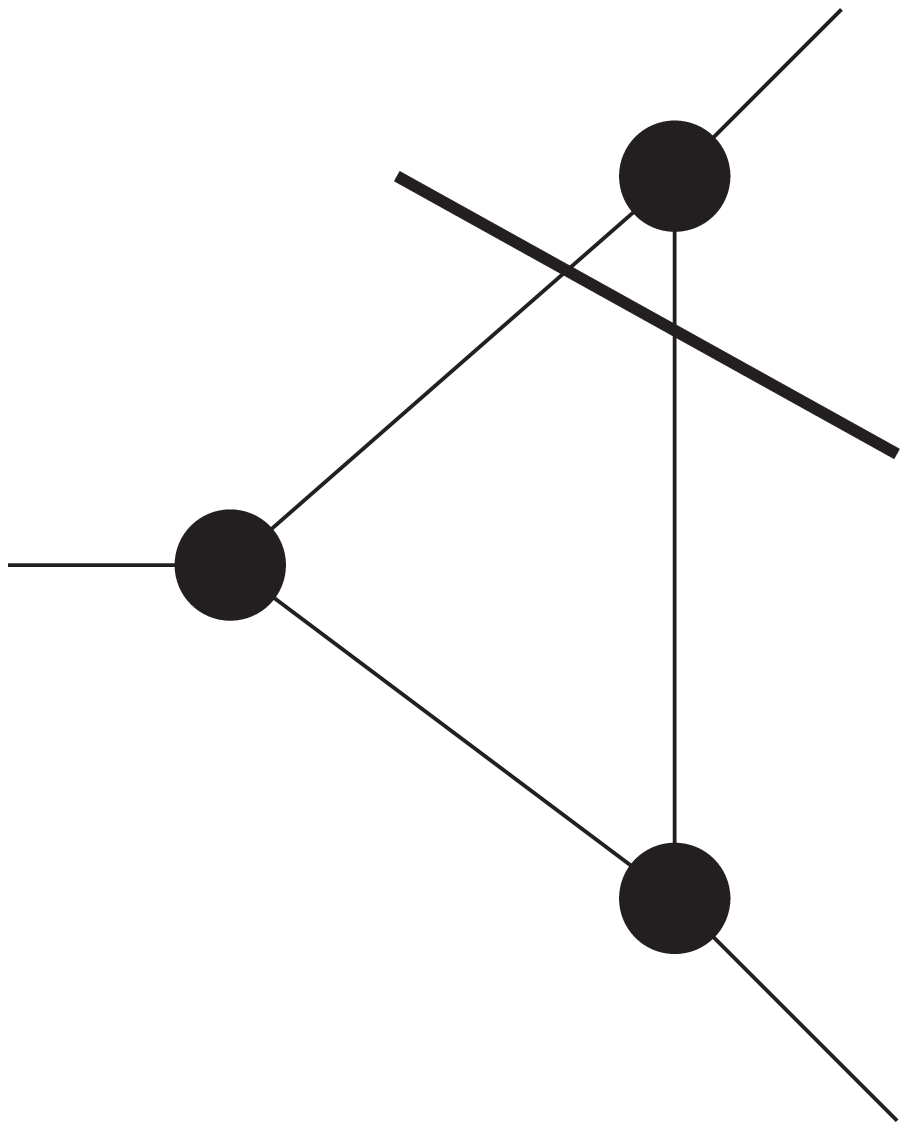}}\;}}
\def\triangleabc{{\;\raisebox{-5mm}{\epsfysize=1cm\epsfbox{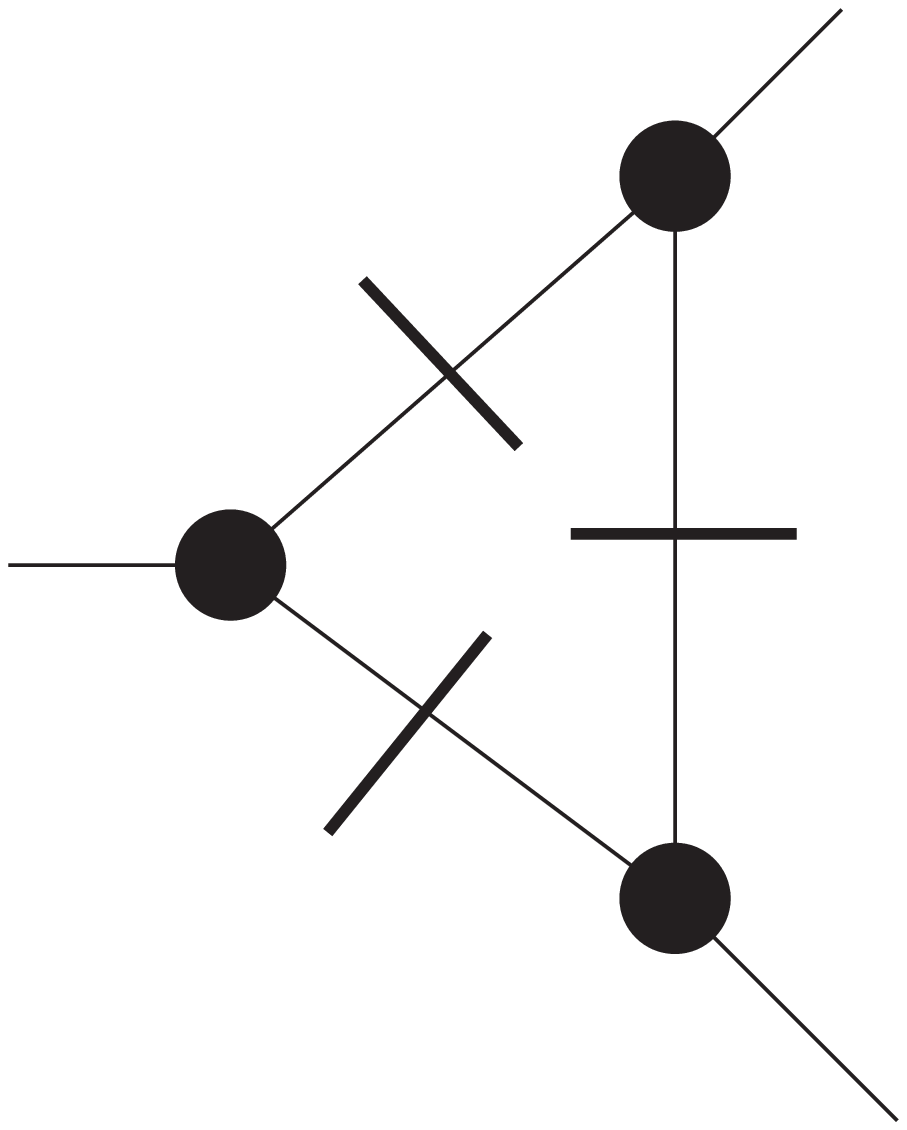}}\;}}

\def\rtra{{\;\raisebox{-5mm}{\epsfysize=1cm\epsfbox{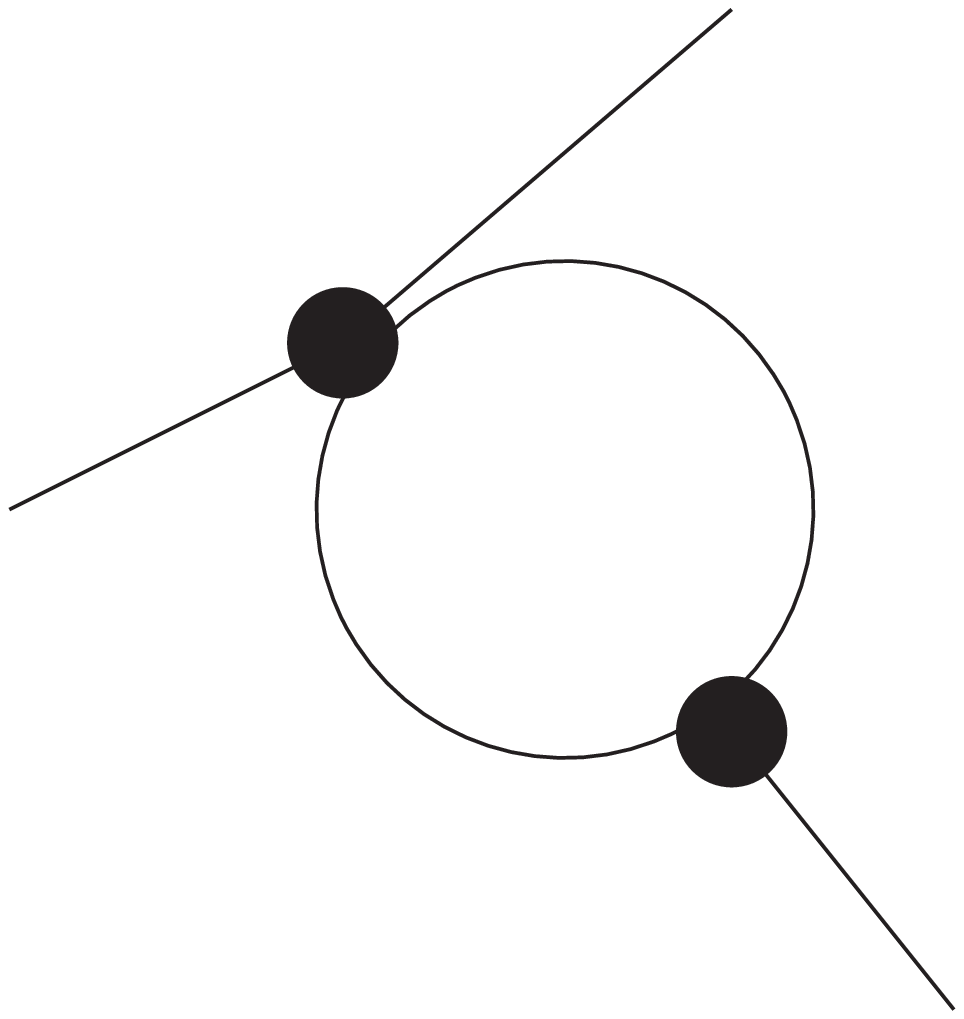}}\;}}
\def\rtrb{{\;\raisebox{-5mm}{\epsfysize=1cm\epsfbox{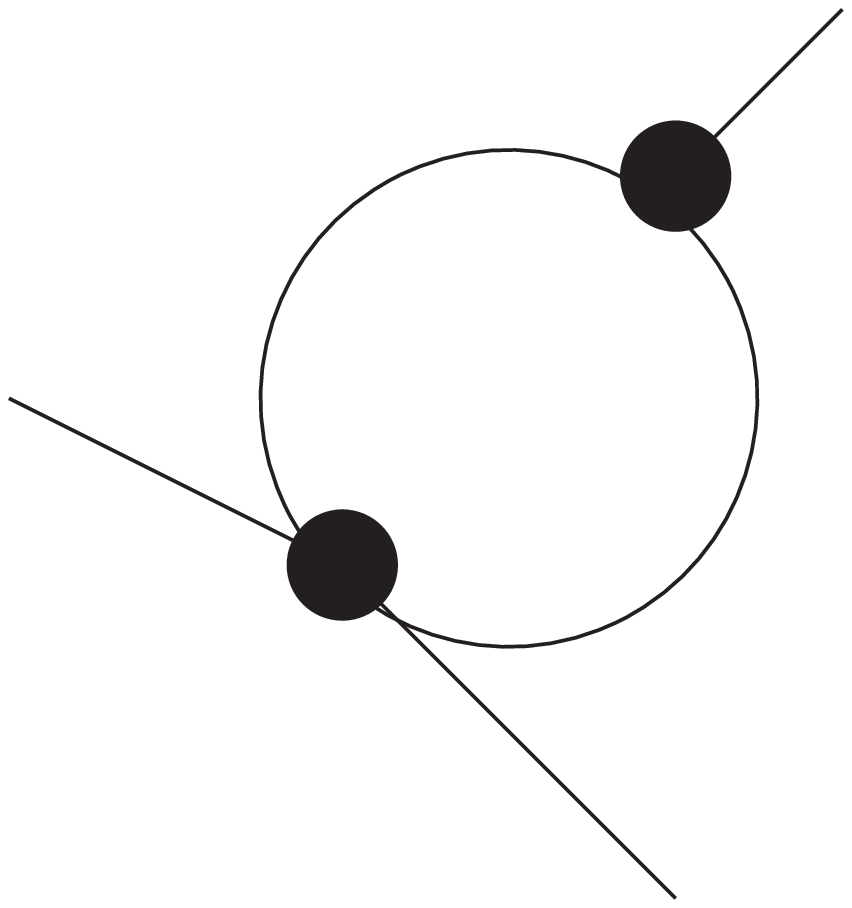}}\;}}
\def\rtrc{{\;\raisebox{-1mm}{\epsfysize=5mm\epsfbox{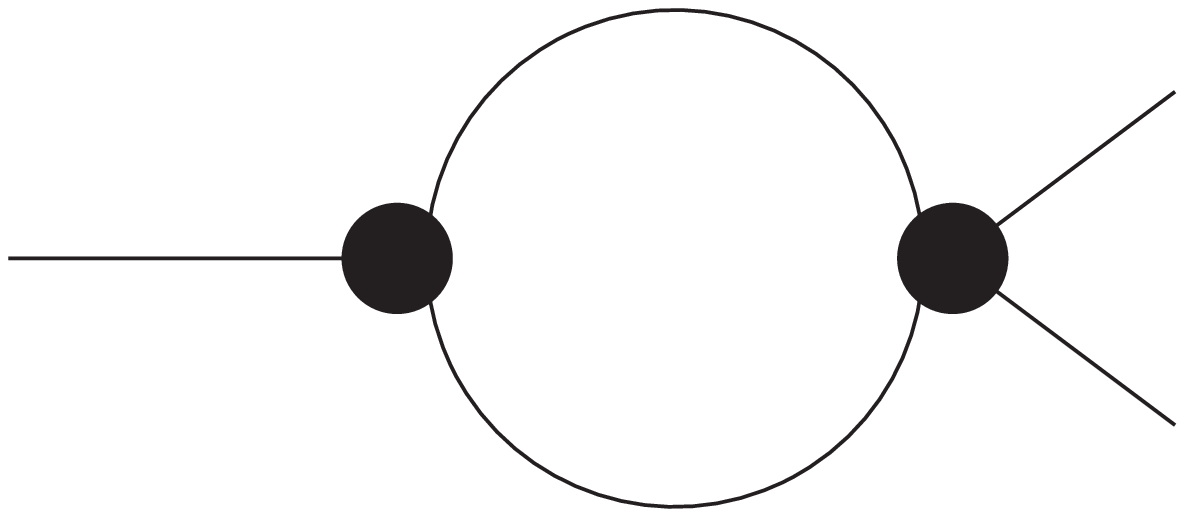}}\;}}

\def\rtrbc{{\;\raisebox{-5mm}{\epsfysize=1cm\epsfbox{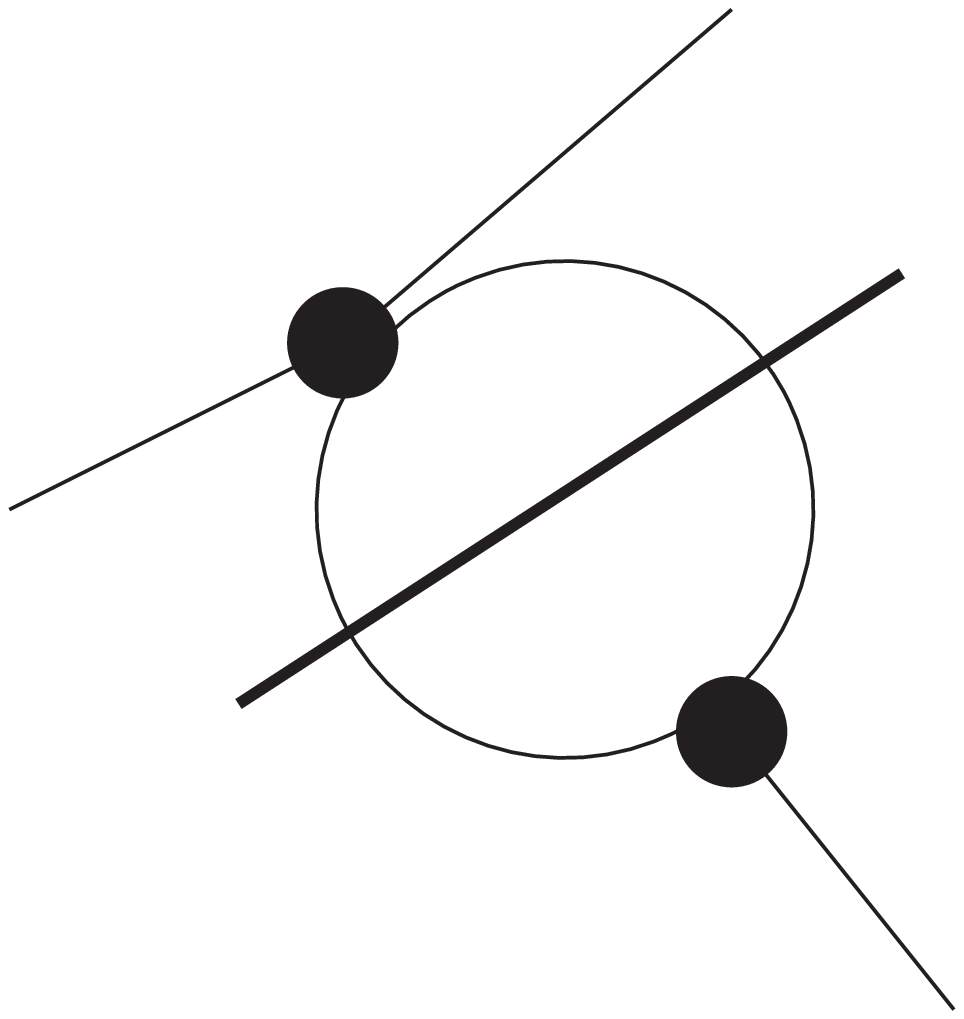}}\;}}
\def\rtrca{{\;\raisebox{-5mm}{\epsfysize=1cm\epsfbox{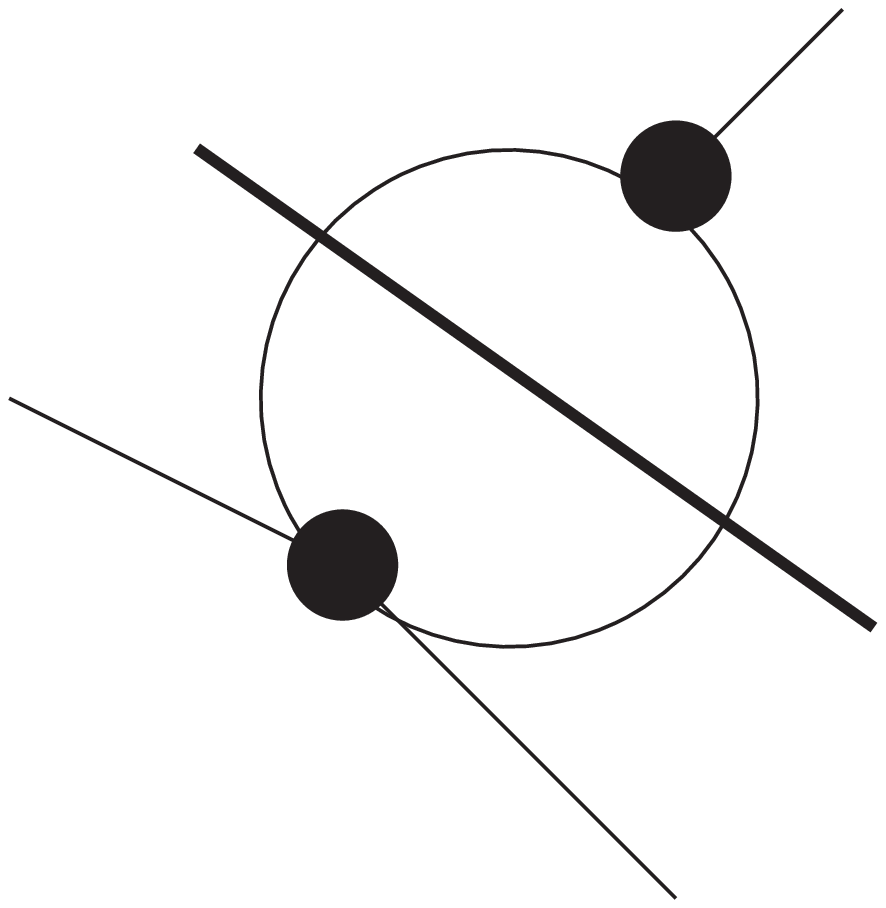}}\;}}
\def\rtrab{{\;\raisebox{-4mm}{\epsfysize=8mm\epsfbox{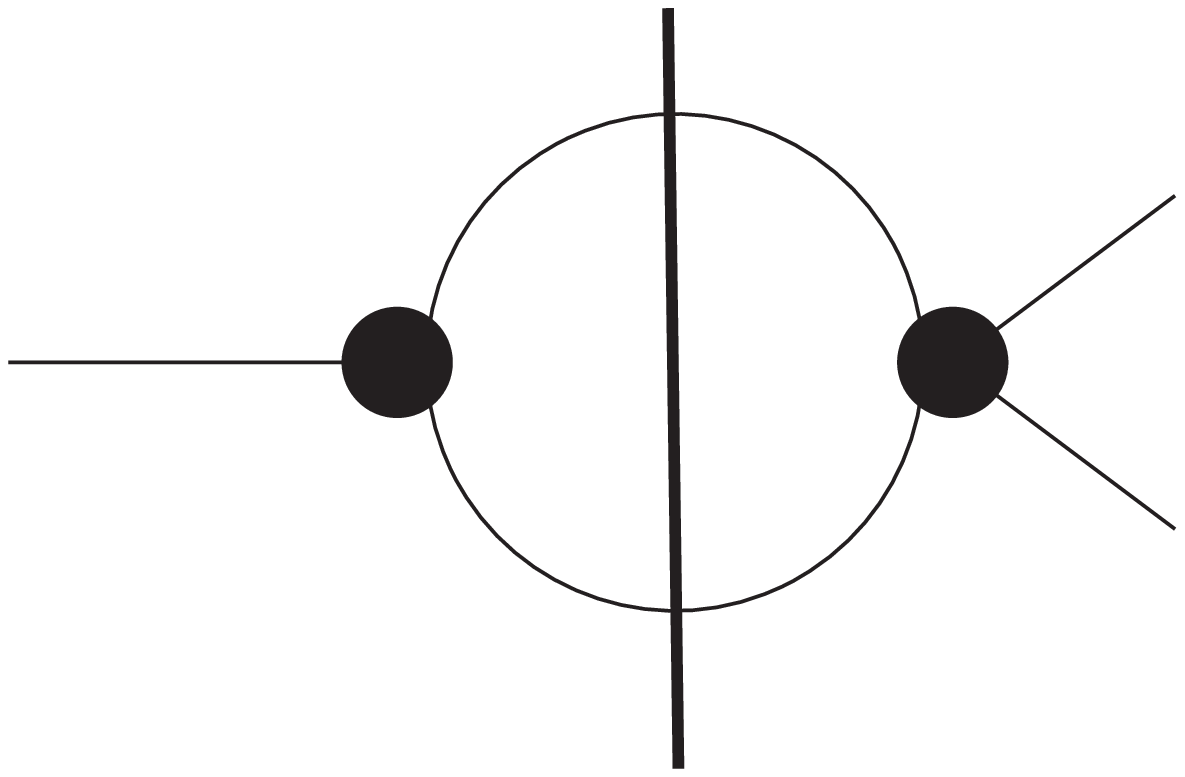}}\;}}

\title{Feynman Amplitudes and Landau Singularities for $1$-loop graphs}
\author{Spencer Bloch}\address{Dept. of Mathematics, University of Chicago, Chicago, IL 60637,
USA\\
E-mail address: bloch@math.uchicago.edu}
\author{Dirk Kreimer}
\address{CNRS-IHES, 91440 Bures sur Yvette, France and Center for Math.\ Phys.\, Boston U., Boston, MA 02215\\ E-mail address: kreimer@ihes.fr}
\begin{document}
\maketitle
\section{Introduction}
The subject of Feynman amplitudes with variable momenta and non-zero masses has been studied by physicists since the 1950's. In the interim, new mathematical methods involving Hodge structures and variations of Hodge structures have been developed. The purpose of this paper is to apply these techniques to the study of amplitudes and Landau singularities in momentum space. While the techniques we develop bear on the general case
here, we will mainly focus on the $1$-loop case.  In this case, for general values of masses and external momenta, the polar locus of the integrand (written in Feynman coordinates) is a smooth quadric. (Exceptionally, in the ``triangle case'', the polar locus is a union of a hyperplane and a quadric.) Mathematically, the polar loci form a degenerating family of such objects, which is a familiar and well-studied situation in algebraic geometry. Our objective is firstly to explain motivically the known fact \cite{DD} that dilogarithms are ubiquitous in this situation, and secondly to show how the motivic and Hodge-theoretic framework is a powerful way to study thresholds and Landau singularities.  

In section \ref{sectHS} we sketch briefly what we will need from the theory of Hodge Structures. The Hodge structures which arise in the context of one loop graphs are quite simple, but it is important to understand how to pass to limits in order to study thresholds in physics. In section \ref{sect2sym} we develop the basic properties of the second Symanzik polynomial which is treated as a quaternionic pfaffian in the sense of E.\ H.\ Moore \cite{T}. The motives we need to study are hypersurfaces defined by a linear combination of the first and second Symanzik polynomials. 

Section \ref{sectdf} develops the basic calculus of differential forms on projective space which is necessary to calculate the de Rham cohomology of our motives. Section \ref{sectpgi} is devoted to the essential technical result, lemma \ref{lem4.3}, which determines the structure of all the $1$-loop motives.  Section \ref{mot} defines the relevant motives. We show (formula \eqref{5.5}) that the weight graded object is a sum of Tate motives $\Q(i)$ for $i=0, -1, -2, -3$.  In section \ref{sectamp} we consider the amplitude itself and show it is a period of a sub Hodge structure (dilogarithm Hodge structure) involving only $\Q(0), \Q(-1), \Q(-2)$.  Completing this chain of ideas, we show in section \ref{sectdl} that the motive of such a dilog Hodge structure is always a sum of dilogs and squares of logarithms (cf. \eqref{2b} below). The argument is variational and uses Griffiths transversality. The authors learned it from \cite{D}. 

Section \ref{sect_tri} discusses the motive of the $1$-loop graph with $3$ edges, the {\it triangle graph}. It turns out that $gr^WH = \Q(0) \oplus \Q(-1)^{5-\nu} \oplus \Q(-2)$ where $\nu$ is the number of masses $m_i=0$. In an appendix, we discuss a duality theorem which is natural mathematically but doesn't have any obvious physical interpretation. Section \ref{sectlim} is a general discussion from a mathematical viewpoint of Landau poles and thresholds, and section \ref{sect_lim2} discusses the limiting mixed Hodge structures associated to various degenerations. Finally, the last section \ref{sect_phy} offers a physical interpretation of the period matrix in the triangle case via Cutkosky rules. 

\section{Hodge Structures}\label{sectHS}
A pure $\Q$-Hodge structure of weight $n$ is a finite dimensional $\Q$-vector space $H_\Q$ together with a decreasing (Hodge) filtration $F^*H_\C$ defined on $H_\C  := H_\Q\otimes \C$. $F^*$ is required to be $n$-opposite to its complex conjugate in the sense that for any $i$
\eq{2.1a}{H_\C \cong F^iH_\C \oplus \overline F^{n+1-i}H_\C,
}
Here $\overline F^j$ is obtained by applying complex conjugation to $F^j$. It is straightforward to check that if we define $H^{i,j} := F^i\cap \overline F^j$, then \eqref{2.1a} is equivalent to the direct sum decomposition 
\eq{}{H_\C = \bigoplus_i H^{i,n-i}.
}
A $\Q$-mixed Hodge structure is a finite dimensional $\Q$-vector space with an increasing filtration (weight filtration) $W_*H$ as well as a Hodge filtration $F^*H_\C$. We require that the induced Hodge filtration on $gr^W_nH_\C$ give $gr^W_nH$ the structure of a pure Hodge structure of weight $n$ for each $n$. 

The only pure Hodge structures of dimension $1$ are the Tate Hodge structures $\Q(n)$. By definition, $\Q(n)$ has weight $-2n$. We have 
\eq{}{F^i\Q(n)_\C = \begin{cases} 0 & i>-n \\
\Q(n)_\C & i\le -n.
\end{cases}
}
In other words, $\Q(n)_\C = H^{-n,-n}(\Q(n)_\C)$. 

A mixed Hodge structure $H$ is called {\it mixed Tate} if  
\eq{}{gr^W_nH = \begin{cases} 0 & n=2m-1 \\
\bigoplus \Q(-m) & n=2m.
\end{cases}
}

The central result of this paper is that Feynman amplitudes at $1$ loop involve only mixed Tate Hodge structures. 
Moreover, the Hodge structures which arise have only $3$ non-trivial weights which we can take to be $0, 2, 4$. 
We refer to them as {\it dilogarithm Hodge structures}. 
\begin{defn}\label{defndl}A dilogarithm mixed Hodge structure $H$ is a mixed Tate Hodge structure 
such that for some integer $n$, we have $gr^W_{2p}H = (0)$ for $p\neq n, n+1, n+2$. 
\end{defn}
We will see in section \ref{sectdl} that periods of dilogarithm Hodge structures have the form
\eq{2b}{\sum_\mu \int \log f_\mu \frac{dg_\mu}{g_\mu}
}
where the $f_\mu, g_\mu$ are rational functions.

For a mixed Tate Hodge structure, the weight and Hodge filtrations are opposite in the sense that
\eq{3}{F^{p+1}H_\C \cap W_{2p}H_\C = (0);\quad H_\C = \bigoplus_{p} (F^pH_\C\cap W_{2p}H_\C).
}
We may choose a basis $\{e^{p,p}_i\}$ of $H_\C$ with $e^{p,p}_i \in F^pH_\C\cap W_{2p}H_\C$. 
\begin{ex}[Kummer extensions] \label{kumex} To an element $x\in \C^\times$ we can associate a mixed Tate Hodge structure $E_x$ with $gr^WE_x = \Q(1)\oplus \Q(0)$. Define a free rank $2$ $\Z$-module $E_{x,\Z} = \Z \ve_{-1}\oplus \Z \ve_0$ with weight filtration $W_{-2}E_{x,\Z} = \Z\ve_{-1} = W_{-1}\subset W_0 = E_{x,\Z}$. Consider the diagram
\eq{2.6b}{\begin{CD} 0 @>>> \Z @> 1\mapsto 2\pi i>> \C @> \exp >> \C^\times @>>> 0 \\
@. @| @A \psi AA @A \ve_{0} \mapsto x AA \\
0 @>>> \Z \ve_{-1} @>>> E_{x,\Z} @>>> \Z \ve_0@>>> 0
\end{CD}
}
Here $\psi(\ve_{-1}) = 2\pi i$ and $\psi(\ve_{0}) = \log(x)$ for some branch of the logarithm. 

By linearity, $\psi$ extends to a $\C$-linear map $\psi: E_{x,\C} = E_{x,\Z}\otimes\C \to \C$. Define
\eq{2.8b}{F^0E_{x,\C} := \ker(\psi_\C) = \C\cdot (\ve_0-\frac{\log x}{2\pi i}\ve_{-1})\subset F^{-1}E_{x,\C} = E_{x,\C}.
}
We take $e^{0,0} = \ve_0-\frac{\log x}{2\pi i}\ve_{-1}$ and $e^{-1,-1} = \frac{1}{2\pi i}\ve_{-1}$. It is traditional for mixed Tate Hodge structures to consider the matrix where the columns, interpreted as coefficients of the $e^{i,i}$ form a basis for the $\Q$-structure. In this case $\ve_0 = e^{0,0} + (\log x)e^{-1,-1},\ \ve_{-1} = 2\pi i e^{-1,-1}$ so the matrix is
\eq{}{\begin{pmatrix}1 && 0 \\ \log x && 2\pi i
\end{pmatrix}.
}
The category of Hodge structures is abelian, and 
\eq{2.10d}{\text{Ext}^1(\Q(-1),\Q) \cong \C^\times;\quad E_x \mapsto x
}
\end{ex}

We remark that the category of Hodge structures has a tensor product. The definitions follow easily from the definition of a tensor product of filtered vector spaces. One has for example $\Q(m)\otimes \Q(n) = \Q(m+n)$. Tensoring with $\Q(n)$ for a suitable $n$, we may if we like arrange that any given mixed Tate Hodge structure has weights $0, 2,\cdots,2r$ for some $r$. 

The central point is that the Betti cohomology of any complex variety (indeed, more generally any diagram of complex algebraic varieties) carries a canonical and functorial Hodge structure. Because Betti groups can be computed using differential forms (de Rham cohomology) our Hodge structures will often have another rational structure coming from algebraic de Rham cohomology. This is useful in physics because it explains the powers of $2\pi i$ occurring in formulas. 
\begin{ex} Consider the Hodge structure $H:= H^1(\P^1-\{0,\infty\},\Q)$. By standard topology this group is one dimensional, dual to the first homology which is spanned by a small circle $S$ around $0$ oriented in a counterclockwise direction. Let $z \in H_\Q$ be a generator with $\langle z,S\rangle=1$. As a Hodge structure, $H = \Q(-1)$. On the other hand, the corresponding de Rham cohomology $H^1_{DR}(\P^1-\{0,\infty\})$ is the $\Q$-vector space defined by the $1$-form $dt/t$ where $t$ is the coordinate on $\P^1$. The pairing with homology is given by integration, and since $\int_S dt/t = 2\pi i$, it follows that $dt/t = 2\pi iz \in H$. The $\Q$-vector space $H$ is the Betti cohomology, $H = H^1_B(\P^1-\{0,\infty\},\Q)$ and the de Rham $\Q$-structure is given in this case by $H_{DR} = 2\pi iH_B$. 
\end{ex}

Families of varieties give rise to families of Hodge structures. Of particular interest is the {\it nilpotent orbit theorem} which is the basic tool in describing degenerations. In the physics surrounding $1$-loop Feynman graphs these degenerations (thresholds) are not as well understood as they might be, and we will show how the nilpotent orbit theorem can be applied. 

To avoid complications, we focus on a $1$-parameter degeneration $\{H_t\}$ parametrized by $t\in D^* = \{t\in \C\ |\ 0<|t|<\ve\}$. This means that we are given a local system $\sH$ over $D^*$ with fibre $H_t$. We have a weight filtration which is an increasing filtration $W_*\sH$ on the local system, and a Hodge filtration which is a decreasing filtration by coherent subbundles $F^*(\sH\otimes_\C \sO_{D^*})$. The point is that the Hodge filtration is not horizontal for the flat structure determined by the local system $\sH$. However, {\it Griffiths transversality} says 
\eq{2.9a}{\frac{d}{dt}F^i(\sH\otimes_\C \sO_{D^*}) \subset F^{i-1}(\sH\otimes_\C \sO_{D^*}).
}
A very general result in algebraic geometry gives that the monodromy on our local system is quasi-unipotent. In other words, replacing $t$ by $u=t^n$ for some $n$, the action $\sigma$ of winding around the puncture in $D^*$  on a fibre of $\sH$ will be unipotent. This allows us to deal for example with the square roots which one faces in 1-loop computations early on. Assuming this has been done, we write 
\eq{2.10b}{N := \log(\sigma)
}
so $N$ is a nilpotent endomorphism of a fibre. 

If we choose $t_0 \in D^*$, we can identify our variation of Hodge structure as a single $\Q$-vector space $H = \sH_{t_0}$ with a weight filtration $W_*H$, a nilpotent endomorphism $N: H \to H$ stabilizing $W_*$, and a variable Hodge filtration $F^*_tH$. In this situation, the nilpotent orbit theorem gives a decreasing filtration $F_{lim}H_\C$ such that the orbit of the one parameter subgroup $\exp(N\frac{\log t}{2\pi i})$ acting on the filtration $F_{lim}H_\C$ approximates the given $F^*_tH$:
\eq{2.11b}{\exp\left(N\frac{\log t}{2\pi i}\right)F_{lim} \sim F_t.
}
Another way to think about \eqref{2.11b} is to note (again by a general result in algebraic geometry) that the coherent sheaf $\sH\otimes \sO_{D^*}$ extends to a coherent sheaf $\widetilde \sH$ on $D$ in such a way that the connection on $\sH$ extends to a connection with log poles on $\widetilde \sH$, i.e. we have
\eq{}{\widetilde\nabla: \widetilde \sH \to \widetilde \sH\cdot \frac{dt}{t}. 
}
Choose a basis $\gamma_i$ for $H^\vee$ which we then view as a multi-valued basis of the local system $\sH^\vee$. We view the $\gamma_i$ as homology classes. For $\omega_t$ a section of $\widetilde\sH$ we write $\langle\gamma_i,\omega_t\rangle = \int_{\gamma_i} \omega_t$ as an integral. Then the entries of
\eq{2.13b}{ \exp\left(-N\frac{\log t}{2\pi i}\right)\begin{pmatrix} \vdots \\
\int_{\gamma_i}\omega_t \\
\vdots\end{pmatrix} 
}
are single-valued functions on $D^*$, and the limit $|t|\to 0$ exists. 

Furthermore, the filtration $F_t$ is meromorphic with respect to the extension in the sense that I can find a basis of the global sections of $\sH\otimes \sO_{D^*}$ which is compatible with the filtration $F^*(\sH\otimes \sO_{D^*})$ and which lies in $t^{-M}\widetilde\sH$ for $M\>>0$. This means there exists a unique saturated filtration $F^*\widetilde\sH$ inducing $F_t$ on $\sH\otimes\sO_{D^*}$. If we choose a basis $\omega_{t,j}$ of $\widetilde\sH$ compatible with the filtration and compute the limits in \eqref{2.13b}, we obtain a concrete matrix representation for $F_{lim}$. Note that $F_{lim}$ depends on the choice of a parameter $t$. For example, if I replace $t$ by $ct$ with $c\in \C^\times$ then the limit in \eqref{2.13b} is multiplied by $\exp(-N\frac{\log c}{2\pi i})$. 

In the context of the limiting Hodge filtration, Griffiths transversality \eqref{2.9a} becomes the condition
\eq{2.14b}{NF_{lim}^i \subset F_{lim}^{i-1}.
}

The limiting filtration $F_{lim}$ is itself the Hodge filtration for the {\it limiting mixed Hodge structure} $H_{lim}$. In general, the weight filtration on $H_{lim}$ is not the limit of the weight filtrations on the $H_t$. For example, the classical situation is when $H_t$ are {\it pure}, i.e. have a single weight. In that case, the limit weight filtration is determined in a canonical way by the nilpotent endomorphism $N$. For applications to $1$-loop amplitudes, we are interested in limits of dilogarithm mixed Tate Hodge structures. The action of monodromy on $gr^WH_t$ will be finite. (This is true quite generally because monodromy will stabilize both an integral and a unitary structure and hence lie in the intersection of a discrete group and a compact group. Such an intersection is necessarily finite.) Typically, in our examples, the eigenvalues of monodromy on $gr^WH_t$ will be $\pm 1$ so it may be necessary to replace $\sigma$ by $\sigma^2$. This explains the presence of $\sqrt{t}$ in formulas found by physicists. Once $\sigma$ is unipotent, however, in our examples, the weight filtration on $H_{lim}$ will be the given weight filtration. 

\begin{ex}\label{ex2.3} We consider a family of Kummer Hodge structures as in example \ref{kumex}. In \eqref{2.6b}, take $\ve_0 \mapsto x(t)$ for $x(t)$ a meromorphic function on the disk $D$, holomorphic away from $0$. Write $x(t) = t^Mu(t)$ with $u(0)\neq 0, \infty$.  With notation as in that example, we take $\gamma_i = \ve_i^\vee, \ i=0,-1$ (dual basis) and $\omega_i = e^{i,i}$. We get
\eq{}{\begin{pmatrix} \int_{\gamma_0}\omega_0 &&  \int_{\gamma_0}\omega_{-1} \\
\int_{\gamma_{-1}}\omega_0 &&  \int_{\gamma_{-1}}\omega_{-1}\end{pmatrix} = \begin{pmatrix}1 && 0 \\
\frac{-\log(t^Mu(t))}{2\pi i} && \frac{1}{2\pi i} \end{pmatrix}
}
The monodromy is given by
\eq{}{N = \begin{pmatrix}0 && 0 \\ -M && 0 \end{pmatrix}.
}
Clearly we should take 
\eq{}{F_{lim}  = \lim_{t\to 0} \exp\begin{pmatrix}0 && 0 \\ +M\log(t)/2\pi i && 0 \end{pmatrix}\begin{pmatrix}1 && 0 \\
\frac{-\log(t^Mu(t))}{2\pi i} && \frac{1}{2\pi i} \end{pmatrix} = \begin{pmatrix}1 && 0 \\ \frac{-\log(u(0))}{2\pi i} && \frac{1}{2\pi i} \end{pmatrix}
}
In this example $F^0= \C\cdot(\ve_0 - \frac{\log(u(0))}{2\pi i} \ve_{-1})= \C e^{0,0}_{lim}$ (defining $e^{0,0}_{lim}$). We take $e^{-1,-1}_{lim} = \frac{\ve_{-1}}{2\pi i}$.

It is a straightforward exercise to extend this construction to mixed Tate variations $H_t$ with $gr^WH_t = \Q(0)^p \oplus \Q(1)^q$ for arbitrary $p,q \ge 1$.
\end{ex}
\begin{ex} Suppose now $gr^WH_t = \Q(0)\oplus \Q(1)\oplus \Q(2)$. We associate to $H_t$ the two Kummer extensions $H_t'=W_{-2}H_t$ and $H_t''=H_t/W_{-4}H_t$. We assume as in example \ref{ex2.3} above that we have calculated the logarithms of monodromy $N', N''$. We can write $H_{lim,\C} =\C e^{0,0}_{lim} \oplus  \C e^{-1,-1}_{lim} \oplus  \C e^{-2,-2}_{lim}$ in such a way that $W_{i,\C} = \sum_{j\le i} \C e^{j,j}_{lim}$ is stable under $N$ and $NF^i_{lim} = N(\sum_{j\ge i} \C e^{j,j}_{lim}) \subset NF^{i-1}_{lim}$.  This implies 
\eq{2.18b}{Ne^{-2,-2}_{lim} = 0;\quad Ne^{-1,-1}_{lim} = a'e^{-2,-2}_{lim};\quad Ne^{0,0}_{lim} = a''e^{-1,-1}_{lim}.
}
Griffiths transversality for $N$ \eqref{2.14b} implies that $Ne^{0,0}_{lim}$ does not involve $e^{-2,-2}_{lim}$. As a consequence, $N$ for the dilogarithm motive is determined by $N'$ and $N''$ for the Kummer sub and quotient motives. These are usually straightforward to calculate. 
\end{ex}

\section{The Second Symanzik Polynomial}\label{sect2sym}

The second Symanzik polynomial is used in the calculation of the
Feynman amplitude associated to a graph $G$ with possibly non-trivial
external momenta. In the physics literature it is usually
derived directly from the linear algebra of Feynman coordinates
\cite{Smirnov}, \cite{Tod}. We will show in this section that it can also be
interpreted as a pfaffian in the   
sense of E.\ H.\ Moore \cite{T} associated to a quaternionic hermitian
matrix. We give the pfaffian construction, but our proof that the
resulting polynomial coincides with the second Symanzik polynomial is
not particularly elegant, so it is relegated to an appendix in the following section. Two
consequences of the pfaffian viewpoint which we do not
pursue further, are firstly that the polynomial is a {\it
  configuration polynomial} for quaternionic subspaces of a based
quaternionic vector space, and hence the techniques of \cite{P} should
apply to the study of the singularities, 
and secondly that for each loop number there is a universal
family. For example, with one loop and $6$ edges, the hypersurface
defined by the second Symanzik is the complement in $\P^5$ of the
complex points of a coset space $GL_2(A)/U_2(A)$ where $A$ is the
quaternions and $U_2(A)$ is the subgroup of $2\times 2$ quaternionic
matrices $M$ satisfying $\overline M^t = M^{-1}$. The Betti numbers
of these coset spaces are known, and one may hope to better understand
the motives of physical interest from this viewpoint.  It will be interesting to study 
the corresponding family at two loops in future work. 

Note that in the presence of non-trivial masses $m_i$, the actual polynomial of
physical interest is 
\eq{2.1}{\Phi(A,q)-(\sum m_i^2A_i)\Psi(A)
}
where $\Psi$ and $\Phi$ are respectively the first and second Symanzik
polynomials.

We write the quaternions $\sA = \R\cdot 1 \oplus \R\cdot i\oplus \R\cdot j \oplus \R\cdot k$ as usual, and we embed $\sA \inj M_2(\C)$ by 
\ga{}{1 \mapsto \begin{pmatrix}1 & 0 \\ 0 & 1\end{pmatrix}; \quad i \mapsto \begin{pmatrix}i & 0 \\ 0 & -i\end{pmatrix} \\
j \mapsto \begin{pmatrix}0 & -1 \\ 1 & 0\end{pmatrix};\quad  k \mapsto \begin{pmatrix}0 & -i \\ -i & 0\end{pmatrix}. \notag
}

Let $u = \begin{pmatrix}0 & -1 \\ 1 & 0\end{pmatrix}$. One checks that the anti-involution $x \mapsto \bar x$ on $\sA$ given by $\bar\ve = -\ve$ for $\ve = i, j, k$ corresponds to $m \mapsto u^{-1}m^t u$ on $M_2(\C)$. More generally, we may embed $M_n(\sA) \inj M_n(M_2(\C))\inj M_{2n}(\C)$ and the anti-involution $x \mapsto \bar x^t$ on $M_n(\sA)$ corresponds to $M \mapsto U^{-1}M^tU$ where $U$ is the diagonal matrix with $u$ along the diagonal. Note that $U$ is skew-symmetric, $U^t = -U$. 

The reduced norm, $Nrd: M_n(\sA) \to \R$ is a polynomial of degree $2n$ which corresponds to the determinant on $M_{2n}(\C)$. 

Let $Herm \subset M_n(\sA)$ be the $\R$-vector space of {\it Hermitian} elements, which we can think of as all elements of the form $x+\bar x^t$. 
\begin{prop}[Moore, Tignol]\label{prop2.1a} There exists a unique polynomial map, the pfaffian norm or Moore determinant $Nrp: Herm \to \R$ such that $Nrp(I) = 1$ and $Nrp(y)^2 = Nrd(y)$. 
\end{prop}
\begin{proof}We can compute in $M_{2n}(\C)$. We have
\ml{}{\det(M + UM^tU^{-1}) = \det((MU-(MU)^t)U^{-1}) = \\
\det(MU-(MU)^t)\det(U^{-1})=(\text{pfaff}(MU-(MU)^t))^2\cdot \text{pfaff}(U^{-1})^2,
}
using the fact that the determinant of a skew matrix is the square of the pfaffian. 
\end{proof}
\begin{cor}\label{cor2.2a}Suppose $M$ in the above proposition is block diagonal with quaternionic hermitian matrices $M_1,\dotsc,M_p$ along the diagonal. Then $Nrp(M) = \prod Nrp(M_j)$. 
\end{cor}
\begin{proof} The assertion is true for the usual pfaffians for skew matrices, and all the matrices in the proof of proposition \ref{prop2.1a} are in block diagonal form. 
\end{proof}
One way to construct elements in $Herm$ is to take $\R$-linear combinations of rank $1$ hermitian elements $x=\bar x^t$. The latter are given by
\eq{}{x=\begin{pmatrix}\bar a_1\\ \bar a_2\\ \cdots \\ \bar a_n \end{pmatrix}\cdot \begin{pmatrix}a_1,a_2\dotsc, a_n \end{pmatrix} = (\bar a_ia_j)_{1\le i, j\le n},
}
where $a_1,\dotsc,a_n \in \sA$. Given a collection $x_1,\dotsc,x_p$ of such hermitian elements, we can construct a polynomial of degree $n$ in $A_1,\dotsc,A_p$ by taking
\eq{2.4a}{\Phi(A_1,\dotsc,A_p) := Nrp(\sum_{i=1}^p A_ix_i).
}

View $\sA^p$ as a right $\sA$-vector space of column vectors. Let $\sH \subset \sA^p$ be a subspace with $\dim_\sA\sH = n$. Choose a basis $\alpha_1,\dotsc,\alpha_n$ for $\sH$ with $\alpha_i=(a_{1i},\dotsc,a_{pi})^t$. Define
\eq{}{e_j^\vee = (a_{j1},\dotsc,a_{jn}),\quad 1\le j\le p
}
Take $x_j = \bar e_j^{\vee,t}\cdot e_j^\vee$ and define  
\eq{2.6a}{\Phi_\sH := Nrp(\sum A_ix_i)
}
as in \eqref{2.4a}. Writing $\alpha = (a_{ij})$, a $p\times n$ matrix, one sees that a different choice of basis for $\sH$ yields a matrix $\beta = (b_{ij})= \alpha M$ where $M$ is $n\times n$ and invertible. We have $(b_{j1},\dotsc,b_{jn}) = e_j^{\vee}M$ so $x_j$ is replaced by $\overline{M}^txM$. 
\begin{lem}\label{lem2.3} Let $M, N$ be $n\times n$ matrices with
  entries in $\sA$. Assume $N = \overline{N}^t$ and $M$ is
  invertible. then  
\eq{2.7a}{Nrp(\overline{M}^tNM) = Nrd(M)Nrp(N). 
}
\end{lem}
\begin{proof}Both sides of \eqref{2.7a} are polynomial maps in the entries of $M$ and $N$, and they have the same square. It follows that the ratio is constant. For $M$ the identity matrix, the ratio is $1$. (Note $Nrd(M) = \det(\iota(M))$ where $\iota: Mat_n(\sA) \inj Mat_{2n}(\C)$ is defined via the embedding $\sA \inj Mat_2(\C)$. In particular, $Nrd(M)=Nrd(-M)$.)
\end{proof}

As a consequence of the lemma, $\Phi(\sum A_i\overline{M}^tx_iM) = Nrd(M)\Phi(\sum A_ix_i)$ so $\Phi_\sH$ is well defined upto a non-zero constant factor.

Consider a graph $\Gamma$ with edge set $E$ and vertex set $V$. Let $H = H_1(\Gamma,\Q)$, and choose a basis $H \cong \Q^r$. We have 
\eq{3.9c}{0 \to \Q^r \to \Q^E \xrightarrow{\partial} \Q^{V,0} \to 0,
}
where $\Q^{V,0} \subset \Q^V$ is the image of the boundary map $\partial$. If we tensor with $\sA$ we get 
\eq{3.10c}{0 \to \sA^r \to \sA^E \to \sA^{V,0} \to 0
}
Suppose we are given $q:= (\ldots q_v,\ldots) \in \sA^{V,0}$. Let $\sH_q \subset \sA^E$ be the sub right $\sA$-module in $\sA^E$ spanned by $\sA^r$ and a lifting $\tilde q$ of $q$. To each $e\in E$ we define an $r+1$-vector $w_e=(w_{e,1},\dotsc,w_{e,r+1})$ by looking at the $e$-th coordinate of the $r$ basis vectors for $H \otimes \sA$ together with $\tilde q$. Note $w_{e,1},\dotsc,w_{e,r} \in \R$. Define (quaternionic) hermitian matrices
\eq{2.10a}{x_e := \overline{w_e}^t\cdot w_e. 
}
The {\it second Symanzik polynomial} is the configuration polynomial \eqref{2.6a} for $\sH=\sH_q$
\eq{2.11a}{\Phi(A)_{\Gamma,q} := Nrp(\sum_E A_ex_e). 
}

\begin{ex}Take $r=1$ and $H=\Q (e_1+\cdots +e_n)$. (This is the $1$-loop case.) Let $\tilde q = \sum \mu_e e \in \sA^E$. Then 
\eq{}{N:= \sum A_ex_e = \begin{pmatrix} \sum_E A_e & \sum A_e\mu_e \\ \sum A_e \bar \mu_e & \sum A_e \bar\mu_e\mu_e\end{pmatrix}.
}
We will see that in this case
\ml{2.13a}{\Phi(A) =Nrp(N) = -(\sum A_e \bar\mu_e)(\sum A_e\mu_e) + (\sum A_e)(\sum A_e\bar\mu_e\mu_e)= \\
\sum_{i<j}\overline{(\mu_i-\mu_j)}(\mu_i-\mu_j)A_{e_i}A_{e_j}.
}
The physics convention  would write $\mu_i = \sum_{j=i}^n q_j$ with $\mu_1=0$. The result in \eqref{2.13a} becomes
\eq{2.14a}{\Phi(A)_{\Gamma,q} = \sum_{i<j}\overline{(q_i+\cdots +q_{j-1})}(q_i+\cdots +q_{j-1})A_iA_j.
}
Again as in \eqref{2.1}, the polynomial of physical interest is
\eq{2.16}{D(q,A):= \sum_{i<j}\overline{(q_i+\cdots +q_{j-1})}(q_i+\cdots
  +q_{j-1})A_iA_j-(\sum m_i^2A_i)(\sum A_i).
}
\end{ex}
\begin{ex}
Let us actually discuss one more example. Consider the three-edge banana.
We take as a basis the two independent cycles $\{e_1,e_2\}$ and $\{e_2,e_3\}$.
The matrix is then given as
$$ 
N:=A_1 (1,0,\bar{\mu_1})^T\cdot (1,0,{\mu_1})
+A_2 (1,1,\bar{\mu_2})^T\cdot (1,1,{\mu_2})
+A_3 (0,1,\bar{\mu_3})^T\cdot (0,1,{\mu_3}),
$$
\eq{}{
N=\begin{pmatrix}
A_1+A_2 & A_2 & A_1\mu_1+A_2\mu_2\\
A_2 & A_2+A_3 & A_2\mu_2+A_3\mu_3\\
A_1\bar{\mu_1}+A_2\bar{\mu_2} &  A_2\bar{\mu_2}+A_3\bar{\mu_3} & A_1\bar{\mu_1}\mu_1+A_2\bar{\mu_2}\mu_2+A_3\bar{\mu_3}\mu_3
\end{pmatrix}
}
We have $NRP(N)=A_1A_2A_3\overline{(\mu_1-\mu_2+\mu_3)}(\mu_1-\mu_2+\mu_3)$.
\end{ex}

\section{Appendix to section \ref{sect2sym}}
It remains to show that our definition of the second Symanzik polynomial coincides with the classical physical definition \cite{ItzZ}, formulas 6-87 and 6-88. (The argument which follows parallels the argument for scalar momenta given in \cite{P}.) 

\begin{lem}Let $\sH \subset \sA^p$ be a subspace as above. Then $\Phi_\sH$ has degree $\le 1$ in each $A_i$. 
\end{lem}
\begin{proof} First note that if $\alpha_1,\dotsc,\alpha_n \in \sA^p$ satisfy a linear relation $\sum \alpha_i a_i=0$, then the $x_i$ in \eqref{2.6a}, viewed as map of row vectors $\sA^n \to \sA^n$ by multiplication on the right, kills the row vector $((\bar a)_1,\dotsc,(\bar a)_n)$. It follows that the matrix $\sum A_ix_i$ does not have maximal rank, so $\Phi(\sum A_ix_i)=0$. 

If some $A_i$ appears to degree $\ge 2$ in some monomial in $\Phi_\sH$, then the monomial can contain at most $\dim \sH-1$ distinct $A_j$. Let $T\subset \{1,\dotsc,p\}$ be the indices occurring in this monomial. By assumption, $\#T<\dim\sH$. Consider the diagram
\eq{2.15a}{\begin{CD} \sH @>>> \sA^p \\
@| @VV\text{proj} V \\
\sH @>\iota >> \sA^T
\end{CD}
}
It is immediate that $\Phi_\sH|_{A_k=0, k\not\in T}$ is the configuration polynomial for the bottom row in \eqref{2.15a}. If this is non-zero, then by the above, the map $\iota$ must be injective. In particular, $\#T\ge \dim \sH$, a contradiction. 
\end{proof}

We now consider $\sH_q \subset \sA^E$ as above. Let $C \subset E$ with $\# C=r+1$. A necessary condition for the monomial $\prod_{e\in C} A_e$ to appear in $\Phi_{\Gamma,q}$ is that $H \inj \Q^E/\Q^{E-C}$. Such a set $C$ of edges is called a {\it cut set}. For a cut set $C$, there exists a spanning tree $T$ and an edge $e\in T$ such that $T-e=E-C$. We choose an $\sA$-basis $h_1,\dotsc,h_{r+1}$ for $\sH_q$ such that $h_1,\dotsc,h_r \in H_1(\Gamma,\Q) \subset \sH_q$. For $c\in C$ let $w_c: \sH_q \to \sA^C$ be the map $w_c(h) = c^\vee(h)c$. Let $\bar w_c^t: \sA^C \to \sH_q$ be the map $\sA^C \surj c\sA \to H\cong \sA^{r+1}$ given by $c\mapsto (c^\vee(h_1),\dotsc,c^\vee(h_{r+1}))$. Here we identify $\sH_q=\sA^{r+1}$ using the basis $\{h_i\}$. Note that $\bar w_c^t w_{d} = 0$ for $c\neq d$. It follows that writing $R_C = \sum_{c\in C} w_c$ we have $\bar R_C^t R_C = \sum_{c\in C} x_c$ with $x_c$ as in \eqref{2.10a}. Thus
\eq{2.16a}{Nrp(\bar R_C^t R_C) = Nrp(\sum A_ex_e)|_{A_e=0,\ e\not\in C,\ A_e=1,\ e\in C} = \text{coefficient of $\prod_{c\in C} A_c$ in $\Phi_{\Gamma,q}$}.
}
It follows from lemma \ref{lem2.3} that this coefficient equals $Nrd(R_C)$. Note that the $(r+1)\times (r+1)$-matrix $R_C$ has real entries except for the last column. Let us define the $\sA$-determinant to be the expansion in the last column
\eq{2.17a}{\det{\!}_\sA(R_C) := (-1)^{r+1}\sum_i(-1)^i\det(R_C^{i,r+1})(R_C)_{i,r+1}
}
where $(R_C)^{i,r+1}$ denotes the minor. Note this matrix has $\R$-coefficients, so the determinant is defined. 
\begin{lem}With notation as above, we have $Nrd(R_C) = (\det_\sA(R_C))\overline{(\det_\sA(R_C))}$.
\end{lem}
\begin{proof}By definition $Nrd(R_C)$ is calculated using the embedding $\sA\inj M_2(\C)$ to view $R_C$ as a $(2r+2)\times (2r+2)$-complex matrix and then taking the usual determinant. In other words, one views $R_C$ as a map $(\C^2)^{r+1} \to (\C^2)^{r+1}$. The assertion is thus clear if the entries of $R_C$ all lie in $\R \subset \sA$. In that case, all the $2\times 2$-matrices are real scalar and we just get the square of the usual determinant. For the general case, it suffices to consider a $2N\times 2N$ complex matrix with entries $2\times 2$ scalar diagonal except for the last two columns. Then one checks that the determinant is computed by interpreting the last two columns as a single column of $N$ $2\times 2$-matrices, expanding as above \eqref{2.17a} and then taking the determinant. (Note that under the embedding $\sA \inj M_2(\C)$ the determinant corresponds to $x\bar x$.) 
\end{proof}

Finally, to identify $Nrp$ with the second Symanzik polynomial, we have to show the expansion \eqref{2.17a} coincides with the usual combinatorial description in terms of cut sets. Fix an orientation and an ordering for the edges of $\Gamma$. Let $C$ be a cut set as above. Let $F_i$, $i=1,2$ be disjoint with $\Gamma-C = F_1\amalg F_2$. Note that one of the $F_i$ may be an isolated vertex. Let $\Gamma/\!/F$ denote the $2$-vertex graph obtained by shrinking the two components of $F \subset \Gamma$ to two (separate) vertices $v_1, v_2$. For $e\in E(\Gamma)$ not an edge of $F$, the image $\bar e$ of $e$ in $\Gamma/\!/F$ is either a loop (tadpole) or has boundary the difference of the two vertices, $\partial e= \pm(v_2-v_1)$. We have also $H_1(\Gamma)\cong H_1(\Gamma/\!/F)$. As above we enumerate the edges $e_1,\dotsc,e_{r+1}$ in $\Gamma-C$. Let $\Gamma_i = (\Gamma/\!/F)/e_i$ be obtained by contracting $e_i$. Then $\det(R_C^{i,r+1})$ is the determinant of the map from $H_1(\Gamma)$ with basis $h_1,\dotsc,h_r$ to $\Z^{E-C-\{e_i\}}$ with basis $e_1,\dotsc,\widehat{e_i}\dotsc,e_{r+1}$. Define
\eq{}{a(i) := \begin{cases}+1 & \partial \bar e_i = v_2-v_1 \\
-1 & \partial \bar e_i = v_1-v_2\\
0 & \partial \bar e_i=0.
\end{cases}
}
The key point then is 
\eq{2.19a}{(-1)^i \det(R_C^{i,r+1}) = a(i)b
}
where $b=\pm 1$ is independent of $i$. This can be seen as follows. Let $W=\bigoplus_1^{r+1}\Q e_i$, . The composition $H_1(\Gamma) \subset W\surj W/\Q e_i$ is an isomorphism. The evident basis $\{e_k, k\neq i\}$ of $W/\Q e_i$ induces a basis of $H_1(\Gamma)$. For two different choices of $i$, say $i_1, i_2$, the determinant of the change of basis matrix is $(-1)^{i_1-i_2}$.  Indeed, writing $\ve= \frac{1}{r+1}\sum e_i \in W$ and letting $\det_1, \det_2\in \det H_1(\Gamma)$ be the exterior powers of the basis vectors for the two bases, one has in $\det W$ that $\ve\wedge\det_1 = (-1)^{i_1-i_2}\ve\wedge\det_2$. (Compare both sides with $e_1\wedge\cdots\wedge e_{r+1}$.)

Finally, we deduce from this and \eqref{2.16a} the classical combinatorial description of the second Symanzik polynomial, viz. the coefficient of $\prod_{e\in C} A_e$ is given by
\eq{}{\Big(\sum_{\partial \bar e_i = v_2-v_1}e_i^\vee(h_{r+1}) - \sum_{\partial \bar e_i = v_1-v_2}e_i^\vee(h_{r+1})\Big)\overline{\Big(\sum_{\partial \bar e_i = v_2-v_1}e_i^\vee(h_{r+1}) - \sum_{\partial \bar e_i = v_1-v_2}e_i^\vee(h_{r+1})\Big)}.
}

\section{Differential Forms on Projective Space}\label{sectdf}

 \label{sec1}
We turn now to the study of motives associated to $1$-loop graphs. We
recall first the structure of differential forms on projective space. Let $\sO = \sO_{\P^n}$ be the sheaf of (algebraic) functions on projective $n$-space, and let $\Omega^i = \bigwedge^i \Omega^1$ denote the sheaf of algebraic differential $i$-forms. Fix a basis $A_0,\dotsc,A_n$ for the linear homogeneous forms on $\P^n$. One has an exact sequence
\eq{1.1}{0 \to \Omega^1 \to \bigoplus_{i=0}^n \sO(-1)dA_i \xrightarrow{p}\sO \to 0
}
(Here the $dA_i$ are just labels for the various summands of the direct sum.) Twisting by $1$, the map $p(1)$ maps $dA_i$ to $A_i \in \Gamma(\P^n, \sO(1))$. For example, $p(2)(A_jdA_i-A_idA_j)=A_jA_i-A_iA_j=0$, so $A_jdA_i-A_idA_j\in \Gamma(\P^n, \Omega^1(2))$. It follows that $dA_i/A_i - dA_j/A_j$ is a (meromorphic) section of $\Omega^1$. 

We are interested in $\Omega^{n-1}$. By standard Koszul algebra we get from \eqref{1.1} an exact sequence 
\eq{1.2}{0 \to \Omega^{n-1}(n-1) \to \bigwedge^{n-1}\Big(\bigoplus_0^n \sO\cdot dA_i\Big) \to \bigwedge^{n-2}\Big(\bigoplus_0^n \sO\Big)(1)
}
the map on the right is given by
\eq{1.3}{dA_{i_1}\wedge\cdots\wedge dA_{i_{n-1}} \mapsto \sum_{j=1}^{n-1} (-1)^{j-1}A_{i_j}dA_{i_1}\wedge\cdots\wedge\widehat{dA_{i_j}}\wedge\cdots \wedge dA_{i_{n-1}}.
}
Again by standard Koszul stuff we have an exact sequence (I have dropped the labels $dA_i$)
\eq{1.4}{ \bigwedge^{n}\Big(\bigoplus_0^n \sO\Big)\to \bigwedge^{n-1}\Big(\bigoplus_0^n \sO\Big)(1) \to \bigwedge^{n-2}\Big(\bigoplus_0^n \sO\Big)(2)
}
where the maps are as in \eqref{1.3}. For $0 \le j\le n$, the section 
\eq{1.5}{\tau_j:= dA_0\wedge \cdots \wedge\widehat{dA_j}\wedge\cdots \wedge dA_n
}
on the left maps to 
\eq{1.6}{\Theta_j := \sum_{i\neq j} \pm A_idA_0\wedge \cdots\wedge\widehat{dA_i}\wedge\cdots \wedge\widehat{dA_j}\wedge\cdots \wedge dA_n 
\in \Gamma(\P^n, \Omega^{n-1}(n)). 
}
(The sign in the sum is $(-1)^i$ for $i<j$ and $(-1)^{i-1}$ for $i>j$.) Treating these expressions as differential forms in the evident way, we have  
\eq{1.7}{d\Theta_j = n\tau_j.
}
In particular, if $F = G/H$ is a ratio of homogeneous polynomials with $\deg G-\deg H = n$, then we compute 
\ml{1.8}{d(\Theta_j/F) = n\tau_j/F - dF\wedge\Theta_j/F^2 = \frac{nF\tau_j - (\sum \partial F/\partial A_k dA_k)\Theta_j}{F^2} = \\
\frac{(nF-\sum_{k\neq j} \partial F/\partial A_k A_k)\tau_j - \partial F/\partial A_j dA_j\wedge \Theta_j}{F^2} = \\
\frac{\partial F/\partial A_j (A_j\tau_j - dA_j\Theta_j)}{F^2} = \frac{(-1)^j\partial F/\partial A_j\Omega_n }{F^2}
}
Here $\Omega_n = \sum (-1)^iA_idA_0\wedge\cdots\wedge \widehat{dA_i}\wedge\cdots\wedge dA_n$. Note that \eqref{1.8} is an identity between meromorphic $n$-forms on $\P^n$. 

Replacing, if necessary, $F$ by a power of $F$, we have proven
\begin{lem}\label{polelem} Let $\omega = \frac{P\Omega_n}{F^p}$ be an $n$-form on $U := \P^n-\{F=0\}$. Assume $G=\sum G_i\frac{\partial F}{\partial A_i}$ lies in the ideal generated by the partial derivatives of $F$. Then we can reduce the order $p$ of pole of $[\omega] \in H^n_{DR}(U)$, i.e. there exists a form $\omega' = \frac{G'\Omega}{F^{p-1}}$ which is cohomologous to $\omega$, $[\omega]=[\omega']$. (Here $H_{DR}$ is algebraic de Rham cohomology calculated using algebraic differential forms. It coincides with Betti cohomology.)
\end{lem}

\section{Complex Poincar\'e Group Invariants}\label{sectpgi}
Feynman integrals, after integration, are fuctions of external momenta. 
If the whole integral transforms as a Lorentz scalar, the integral is a function of 
Lorentz invariant scalar products of external momenta. The number of and type of these invariants 
are exhibited here from a mathematical viewpoint, incorporating momentum conservation and the finite dimension of spacetime.

The fact that the amplitudes for $1$-loop graphs are dilogarithms is a
consequence of some basic facts about the invariants of the orthogonal
group. Let $O_\C(r)$ be the subgroup of $GL(\C^r)$ leaving invariant a non-degenerate inner product $(p,q)\mapsto p\cdot q$. Let $G = \C^r \rtimes O_\C(r)$ be the ``complex Poincar\'e group'' generated by orthogonal transformations and translations. As an algebraic group over $\C$, $G$ has dimension $r+\frac{r(r-1)}{2}$. 
Let $G$ act diagonally on $(\C^r)^{r+2}$. The quotient  $(\C^r)^{r+2}/G$ has dimension 
\eq{4.1}{\dim_\C (\C^r)^{r+2}/G = r(r+2) - (r+\frac{r(r-1)}{2}) = \binom{r+2}{2}-1.
}

Let $p_j:(\C^r)^{r+2} \to \C^r, 1\le j\le r+2$ be the projections. Following physics notation we write
\eq{}{(p_j-p_k)^2 := (p_j-p_k)\cdot(p_j-p_k)
}
with the inner product as above. We obtain in this way $\binom{r+2}{2}$ $G$-invariant functions on $(\C^r)^{r+2}$. It follows from \eqref{4.1} that there is an algebraic relation between these functions. 

To understand this relation, we change bases in $\C^r$ so the inner product is the sum of squares of coordinates (Euclidean inner product). We can view $P:=(p_1,\dotsc,p_{r+2})$ as an $r\times (r+2)$ matrix. The $(r+2)\times (r+2)$-symmetric matrix
\eq{4.2}{N:= (p_j\cdot p_k) = P^tP
}
has rank $\le r$. It is convenient at the point to introduce masses $m_j, 1\le j\le r+2$. Consider the $(r+2)\times (r+2)$-symmetric matrix
\eq{4.4}{M(m) := (m_i^2+m_j^2 + (p_i-p_j)^2) = (m_i^2+p_i^2)_{ij} + (m_j^2+p_j^2)_{ij} -2N = M_1+M_2-2N.
}
View $\C^{r+2}$ as column vectors, and let $H \subset \C^{r+2}$ be the codimension $1$ subspace defined by setting the sum of the coordinates to zero. Note that $M_1$ has all columns the same, so for $h\in H$ we have $M_1h=0$. Similarly $h^tM_2=0$. It follows that the quadratic form given by the symmetric matrix $M(m)$ is necessarily degenerate when restricted to $H$, i.e. $\exists\ 0\neq k \in H$ with $h^tM(m)k=0$ for all $h\in H$. 
\begin{lem}For general values of the $p_i$ we have $\det(M(0))\neq 0$.
\end{lem}
\begin{proof}Take $p_1,\dotsc,p_r$ to be the usual orthonormal basis of $\C^r$, and take $p_{r+1}=0$. One easily checks in this case that the coefficient of $p_{r+2,1}^4$ in $\det(M(0))$ is plus or minus a power of $2$. In particular, it is non-zero, and the lemma follows. 
\end{proof}
\begin{rem} Of course, it follows from the lemma that $\det(M(m)) \neq 0$ for general $m$ and $p$ as well. 
\end{rem}

Assume now that $M(m)$ is invertible. Write $\vec{1} = (1,\dotsc,1)\in \C^{r+2}$. It follows from the above that $M(m)k = \kappa\vec{1};\quad \kappa \neq 0$. Scaling $k$, we may assume $k=M(m)^{-1}\vec{1}$. Thus
\eq{4.5}{(M(m)^{-1}\vec{1})\cdot \vec{1} = 0.
}
When the masses are zero, \eqref{4.5} yields the non-trivial algebraic relation between the $(p_i-p_j)^2$. We will interpret \eqref{4.5} in the case $r=4$ as determining where in the weight filtration of a Hodge structure the Feynman integrand lies. In physics terms, it is the statement that for $1$-loop graphs, the amplitude is expressed in terms of logarithms and dilogarithms of Lorentz-invariant rational functions of momenta, \cite{DD}, \cite{Tod}. 

In physical situations, of course, the $p_i$ are $4$-vectors. 

\begin{lem}\label{lem4.3} Fix $n\ge 6$. Let $p_i \in \C^4,\ 1\le i\le n$, and let $m_i\in \C,\ 1\le i\le n$. Let $H \subset \C^n$ be the codimension $1$ linear subspace defined by setting the sum of the coordinates to $0$. The matrix $M(m) = (m_i^2+m_j^2+(p_i-p_j)^2)_{ij}$ has rank $\le 6$. For general values of $m_i, p_i$ the rank is exactly $6$ and the vector $(1,\dotsc,1)$ lies in the image $M(m)(H) \subset \C^n$.  
\end{lem}
\begin{proof}As in \eqref{4.4} $M(m)$ is a sum of three matrices. The matrices $M_1, M_2$ have rank $1$. The matrix $N$ has rank $4$ (for general $p_i$) as in \eqref{4.2}. It follows that $M(m)$ has rank $\le 6$, and it is easy to see the rank is exactly $6$ for general values of the parameters. To show the vector $(1,\dotsc,1)\in M(m)(H)$, it suffices to solve the equations
\eq{4.6}{\sum_{i=1}^n a_i(p_i^2+m_i^2) = 1;\quad \sum_{i=1}^n a_ip_i=0;\quad \sum_{i=1}^n a_i=0.
}
These equations clearly admit a solution in the $a_i$ for general values of the parameters when $n\ge 6$. 
\end{proof}

\section{The Motive}\label{mot}

Let $X:Q=0$ be a rank $\min(6,n+1)$ quadric in $\P^n$. Let
$A_0,\dotsc,A_n$ be homogeneous coordinates, and write $\Delta:\prod
A_i=0$ for the reference simplex. We will be interested in the
''motive'' (or more concretely, the Hodge structure)  
\eq{5.1}{H^n(\P^n-X,\Delta-X\cap\Delta,\Q). 
}
(In the case $n=2$, the triangle graph, the motive of physical interest is slightly
different. We treat it separately in section \ref{sect_tri}.)

We assume that $X$ is in good position with respect to $\Delta$ in the sense that for any face $F\cong \P^{i} \subset \Delta$ the intersection $X\cap F$ has rank $\min(6,i+1)$. In particular, if $\dim F<6$ then $X\cap F$ is smooth. (The nullspace $L\subset X$ is a linear space of dimension $n-6$, and our assumption is that $L$ meets all faces of $\Delta$ properly.)
\begin{lem}\label{lem5.1} (i) We have
\eq{}{H^n(\P^n-X,\Q) \cong \begin{cases}0 & n>5 \\ \Q(-m-1) & n=2m+1 \le 5 \\ 0 & n=2m>0 \\ \Q(0) & n=0
\end{cases}
}
\noindent (ii) $H^k(\P^n-X,\Q) = (0)$ if $0<k\neq n\le 5$ or if $n>5$ and $k\neq 0, 5$. 
\end{lem}
\begin{proof}Suppose first $n> 5$. Let $p:\P^n-L \to \P^5$ be the projection with center $L$. We have $X-L = p^{-1}(Y)$, where $Y \subset \P^5$ is a smooth quadric. It follows that $\P^n-X$ is a fibre bundle over $\P^5-Y$ with fibre $\A^{n-5}$. A standard result for fibrations with contractible fibres yields $H^*(\P^5-Y,\Q) \cong H^*(\P^n-X,\Q)$. Since $\P^5-Y$ is affine of dimension $5$, it has cohomological dimension $5$ so $H^n(\P^n-X,\Q)=(0)$.  

Quite generally, for $X$ a smooth hypersurface in $\P^n$, the Gysin sequence (given by residues on differential forms) yields an exact sequence which reads in part
\eq{}{0 \to H^n(\P^n-X,\Q) \to H^{n-1}(X,\Q(-1)) \to H^{n-1}\P^n,\Q) \to 0.
}
We now assume $n\le 5$ so $X$ is a smooth quadric.  The middle dimensional cohomology of a smooth quadric of dimension $d$ is known to be rank $2$ generated by algebraic cycles for $d$ even and zero for $d$ odd. The lemma follows. 
\end{proof}

For an index set $I = \{i_0,\dotsc,i_p\} \subset \{0,\dotsc,n\}$ write $|I|=p+1$ and let $\Delta_I \subset\Delta\subset \P^n$ be defined by the vanishing of the homogeneous coordinates $A_{i_j}$. The motive \eqref{5.1} is the hypercohomology of the complex of sheaves 
\eq{}{\Q_{\P^n-X} \to \bigoplus_{|I|=1}\Q_{\Delta_I-X\cap \Delta_I} \to \cdots \to \bigoplus_{|I|=n}\Q_{\Delta_I}
}
There is a spectral sequence $E_1^{p,q} = H^q(\bigoplus_{|I|=p}\Q_{\Delta_I-X\cap \Delta_I}) \Rightarrow H^{p+q}(\P^n-X,\Delta-X\cap\Delta,\Q).$ (For simplicity we write $\P^n-X = \bigoplus_{|I|=0}\Delta_I-X\cap \Delta_I$.) 

Suppose first $n\le  5$. The differentials $d_1^{p,q}: E_1^{p,q} \to E_1^{p+1,q}$ are zero in this case except for $d_1^{4,0}: \bigoplus_{|I|=4}H^0(\Delta_I-X\cap \Delta_I,\Q) \to \bigoplus_{|I|=5}H^0(\Delta_I-X\cap \Delta_I,\Q)$ which is simply restriction from $1$-simplices to $0$-simplices. It follows that the weight graded cohomology in these cases is 
\eq{}{gr^WH^n(\P^n-X,\Delta-X\cap \Delta,\Q) = \begin{cases}\Q(0) \oplus\bigoplus_{15}  \Q(-1)  \oplus\bigoplus_{15} \Q(-2)\oplus \Q(-3) & n=5 \\ \Q(0) \oplus\bigoplus_{10}  \Q(-1)  \oplus\bigoplus_{5} \Q(-2) & n=4 \\ \Q(0) \oplus\bigoplus_{6}  \Q(-1)  \oplus \Q(-2) & n=3 \\
\Q(0) \oplus\bigoplus_{3}  \Q(-1) & n=2 \\
\Q(0) \oplus \Q(-1) & n=1 \end{cases} 
}

For $n\ge 6$ the differential $d_1^{n-6,5}:E_1^{n-6,5} \to E_1^{n-5,5}$ is non-trivial. One finds for the weight graded 
\eq{5.5}{gr^WH^n(\P^n-X,\Delta-X\cap\Delta,\Q) = \Q(0)\oplus\bigoplus_{\binom{n+1}{n-1}}\Q(-1)\oplus\bigoplus_{\binom{n+1}{n-3}}\Q(-2)\bigoplus_{c_n}\Q(-3).
} 
Here $c_n$ is the dimension of $\text{coker}(\bigoplus_{|I|=n-6}\Q
\xrightarrow{\partial}\bigoplus_{|I|=n-5}\Q) $. In fact, the weight
$6$ part of these motives will not play a role in our amplitude
calculations. This is because (as we will see in proposition
\ref{prop8.2}) the differential form given by the Feynman integrand
\eqref{6.1} below lies in $W_4H^n(\P^n-X,\Delta-X\cap \Delta,\Q)$.  

\section{The amplitude}\label{sectamp}

Associated to a $1$-loop graph with $n$ internal edges and incoming
momenta ($4$-vectors summing to $0$) $p_i$ at the vertices we have the
second Symanzik polynomial $D(p,A)$ \eqref{2.16} which is a
homogeneous quadric in the variables $A_1,\dotsc,A_n$. The associated
amplitude is 
\eq{6.1}{\int_{\sigma} \frac{(\sum A_i)^{n-4}\Omega_{n-1}}{D(p,A)^{n-2}}.
} 
Here the first Symanzik polynomial is just $\sum A_i$, and
$\Omega_{n-1}$ is as in section \ref{sectdf}. Note if $n \le 3$ then
$\sum A_i$ appears in the denominator. We will focus on the case
$n\ge 4$, leaving the {\it triangle graph} case $n=3$ (we are now counting edges from $1,\ldots,n$,
not from $0,\ldots,n-1$, as in the previous section) to section
\ref{sect_tri}.   

\begin{lem}\label{prop6.1} Assume $n\ge 5$. Let $\P^{n-2}\cong \Delta_i,\ 0\le i\le n-1 \subset \P^{n-1}$, be the maximal faces of the coordinate simplex $\Delta \subset \P^{n-1}$. Let $X: D(p,A) \subset \P^{n-1}$ be the quadric. Assume momenta and masses are general. Then the form $\eta_{n-1}:= \frac{(\sum A_i)^{n-4}\Omega_{n-1}}{D(p,A)^{n-2}}$ on $\P^{n-1}-X$ is exact. we can find an $(n-2)$-form $w_{n-1}$ on $\P^{n-1}-X$ and constants $a_j \in \C$ such that (i) $dw_{n-1} = \eta_{n-1}$; (ii) $w_{n-1}|\Delta_j= \pm a_j\eta_{n-2}$; (iii) $\sum a_j = 0$. 
\end{lem}
\begin{proof}Let $M(m) = (m_i^2+m_j^2+(p_i-p_j)^2)_{1\le i,j\le n}$ be the symmetric matrix corresponding to $D(p,A)$. From lemma \ref{lem4.3} there exists  a column vector $\vec{a} = (a_1,\dotsc,a_n)$ such that $M(m)\vec{a} = (1,\dotsc,1)$ and $\sum a_i=0$. Define
\eq{}{w_{n-1}:= \frac{(\sum A_i)^{n-5}\sum_j (-1)^ja_j\Theta_j}{2(n-3)D(p,A)^{n-3}}
}
where $\Theta_j$ is as in \eqref{1.6}. Using \eqref{1.8} with $F=\frac{D(p,A)^{n-3}}{(\sum A_i)^{n-5}}$ we compute
\ml{}{dw_{n-1} = \sum_j a_j\Big( \frac{(n-3)(\sum A_k)^{n-5}\frac{\partial D}{\partial A_j}}{2(n-3)D^{n-2}} - \frac{(n-5)(\sum A_k)^{n-6}}{2(n-3)D^{n-3}}\Big) = \\
\frac{(n-3)(\sum A_k)^{n-5}\sum_j a_j\frac{\partial D}{\partial A_j}}{2(n-3)D^{n-2}} = \frac{(\sum A_k)^{n-4}}{D^{n-2}} = \eta_{n-1}.
} 
Note finally that $\Theta_j|\Delta_k = \pm \delta_{jk}\Omega_{n-2}$, proving (ii). \end{proof}
\begin{prop}\label{prop8.2} With notation as in the lemma, we have $\eta_{n-1} \in W_4H^{n-1}(\P^{n-1}-X, \Delta-X\cap \Delta)$. The Feynman amplitude for any $1$-loop graph is a period of a dilogarithm mixed Hodge structure as in definition \ref{defndl}.
\end{prop}
\begin{proof}
If $n \ge 6$, the faces $\Delta_j \cong \P^{n-2}$ have dimension $\ge 4$ and we can apply the lemma again to the forms $w_{n-1}|\Delta_j = \pm a_j\eta_{n-2}$. In this way we can build a sort of cascade
\eq{}{\begin{array}{ccccccc}
 &&&& \Omega^{n-2}_{\P^{n-1}-X}& \xrightarrow{d} & \Omega^{n-1}_{\P^{n-1}-X} \\
 &&&& \downarrow \\
 && \bigoplus_i \Omega^{n-3}_{\Delta_i - X\cap\Delta_i} &  \xrightarrow{d} &  \bigoplus_i \Omega^{n-2}_{\Delta_i - X\cap\Delta_i} \\
 && \downarrow  \\
 && \vdots \\
 \bigoplus_{|I|=n-5} \Omega^3_{\Delta_I-X\cap \Delta_I}& \xrightarrow{d} & \ldots \\
 \downarrow \\
 \bigoplus_{|I|=n-4} \Omega^3_{\Delta_I-X\cap \Delta_I}
 \end{array}
}
where the vertical maps are restrictions on faces (with appropriate signs). (We simplify notation by writing $\Omega^i_Z$ for the sections of the sheaf $\Omega^i$ over $Z$ rather than the sheaf itself.) What this means is that the de Rham cohomology of our motive, $H^{n-1}_{DR}(\P^{n-1}-X,\Delta-X\cap \Delta)$ is calculated by a double complex of algebraic differential forms $C^{a,b} = \bigoplus_{|I|=a}\Omega^{b}_{\Delta_I-X\cap\Delta_I}$. The differential $d': C^{a,b} \to C^{a+1,b}$ (resp. $d'': C^{a,b} \to C^{a,b+1}$) is given by restriction to faces of $\Delta$ with appropriate signs (resp. exterior differentiation.) The total differential $d=d'+d''$. We have
\eq{6.5}{H^{n-1}_{DR}(\P^{n-1}-X,\Delta-X\cap \Delta) = H^{n-1}(C^{**},d) = \Big(\bigoplus_{a+b=n-1}C^{a,b}\Big)\Big/ d\Big(\bigoplus_{a+b=n-2}C^{a,b}\Big).
}
(Note that in total degree $n-1$, all cochains are closed.) The cochain $(0,\dotsc,0,\eta_{n-1}) \in C^{0,n-1}$ represents a de Rham class whose period integrated against the homology chain given by $\sigma_{n-1} = \{(a_1,\dotsc,a_n)\ |\ a_i \ge 0, \forall i\}$ is the Feynman amplitude. The content of proposition \ref{prop6.1} is that we can construct a form $w \in \bigoplus_{a+b=n-2}C^{a,b}$ such that 
\eq{6.6}{(0,\dotsc,0,\eta_{n-1}) - dw \in C^{n-4,3}\oplus  C^{n-2,1}\oplus C^{n-1,0}
}
(Note there is no contribution from $C^{n-3,2}$. This is because
$H^{2k}(\P^{2k}-X)=(0)$ for a smooth quadric in even dimensional
projective space of any dimension. The argument is the same as in
lemma \ref{prop6.1}.)   

Finally, It follows from lemma \ref{lem5.1} that the filtration in
  equation \eqref{6.5} coming from the filtration $W_rC^{**} =
  \bigoplus_{a,b;\ a\ge n-2-r}C^{a,b}$ is the weight filtration
  $W_rH^{n-1}_{DR}(\P^{n-1}-X,\Delta-X\cap \Delta)$. 
\end{proof}

Note that this is a proof of an old result of Nickel, who first studied the dependences between 
one-loop graphs in a fixed dimension \cite{Nickel}.
\section{Dilogarithm Motives}\label{sectdl}

We have seen \eqref{5.5}, \eqref{6.6} (note also the last comment in
section \ref{mot}) that motives $H$ arising from $1$-loop amplitudes
satisfy $gr^WH = \Q(0) \oplus \Q(-1)^b \oplus \Q(-2)^c$. They are
mixed Tate motives with weights $0, 2, 4$. In this section we show how
periods of such motives are related 
to dilogarithms. A general reference is \cite{D2}. We will follow the
standard convention and trivialize the one-dimensional vector
space $\Q(n)=\Q$ in such a way
that the Betti structure is $(2\pi i)^n\Q$ so the $DR$-structure is
$\Q$. 

First let us reduce to the case $c=1$. We assume for simplicity that
our de Rham structure is defined over $\Q$. (If not, one need simply
extend the field of coefficients of $H$.) The quotient pure Hodge
structure $H/W_2H \cong \bigoplus \Q(-2)$ satisfies
\eq{}{(H/W_2H)_{DR} = (2\pi i)^2(H/W_2H)_{B} \subset (H/W_2H)_{\C}. 
}

Further, the Hodge filtration has a single non-trivial piece in degree
$-2$ and hence is defined already over $\Q$. What this means is that
we can take our Feynman integrand $\eta$ which we view as lying in
$H_{DR}$ and project it to $(H/W_2H)_{DR}$. The $\C$-line spanned by
this image is canonically identified with $\Q(-2)_\C$, where
$\Q(-2)_{DR} = \Q\cdot\eta$ and $\Q(-2)_B = \Q\cdot(\eta/(2\pi
i)^2)$. The preimage $H'\subset H$ of this copy of $\Q(-2)$ has
weight graded $gr^WH' = \Q(0)\oplus \Q(-1)^b \oplus Q(-2)$, and it suffices
to compute the periods for this Hodge structure. 

Because $H$ is a mixed Tate Hodge structure, there will exist a 
base $e_{-2}, e_{-1,\mu}, e_0$ of $H_\C$ ($1\le \mu \le b$) such that the weight
(resp. Hodge) filtration on $H_\C=\C^{[-2,0]}$ is given by
$\C^{[-i,0]}=W_{2i}H_\C$ and $F^jH_\C = \C^{[-2,-j]}$, and such that
the trivialization given by the $e$'s identifies $gr^WH=\Q(0)\oplus \Q(-1)^b
\oplus \Q(-2)$. Here 
$\C^{[r,s]}$ is the span of the $e_{p,q}$ with $r\le p\le s$. 

We consider first the sub Hodge structure $W_2H$ and the quotient
$H/W_0H$, which are mixed Tate with weights $0, 2$ (resp. $2,
4$). There will exist $b$-tuples $(f_1,\dotsc,f_b)$ and
$(g_1,\dotsc,g_b)$ in $\C^{\times,b}$ such that the Betti structures
on $W_2H_\C = \C e_0\oplus \bigoplus_{\mu=1}^b \C e_{-1,\mu}$
(resp. $(H/W_0H)_\C = \bigoplus_{\mu=1}^b \C e_{-1,\mu}\oplus \C
e_{-2}$) are given by the $\Q$-spans of the elements $\frac{1}{2\pi i}(e_{-1,\mu} + \log
  g_\mu e_0)$ 
(resp. $\frac{1}{(2\pi i)^2}(e_{-2}+\log f_\mu e_{-1,\mu}$). The Betti structure on $H$ will then be the $\Q$-span of the columns of a matrix
\eq{8.2}{\frac{1}{(2\pi i)^2}\begin{pmatrix}1 & 0 & 0 & \hdots & 0 & 0 \\ \log f_1 & 2\pi i & 0 & \hdots & 0 & 0\\ \log f_2 & 0 & 2\pi i & \hdots & 0 & 0\\ \vdots & \vdots & \vdots & \hdots & 0 & 0\\ \log f_b & 0 & 0 & \hdots & 2\pi i & 0\\
h & 2\pi i\log g_b & 2\pi i\log g_{b-1} & \hdots & \log g_1 & (2\pi i)^2
\end{pmatrix},
}
and the challenge is to compute $h$. 

We can compute $h$ upto a constant by using {\it Griffiths transversality}. We treat the $f_\mu, g_\mu, h$ as functions and consider the variation of Hodge structure given by \eqref{8.2}. It carries a connection $\nabla$ for which the columns are horizontal. The transversality condition says $\nabla(F^i) \subset F^{i-1}$ where $F^*$ is the Hodge filtration. 

Write $C_i$ for the columns in \eqref{8.2}. We have
\eq{8.3}{e_{-2} = C_1-\sum_{\mu}\frac{\log f_\mu}{2\pi i}C_{\mu+1} - \frac{1}{(2\pi i)^2}(h-\sum \log f_\mu \log g_\mu)C_{b+2}.
}
Transversality says that 
\ml{}{\nabla e_{-2} = Ae_{-2}+Be_{-1} = -\frac{1}{2\pi i}\sum \frac{df_\mu}{f_\mu}C_{\mu+1} - \frac{1}{(2\pi i)^2}(dh-\sum (f_\mu dg_\mu + g_\mu df_\mu))C_{b+2} = \\
 -\frac{1}{2\pi i}\sum \frac{df_\mu}{f_\mu}e_{-1,\mu} - (dh-\sum (\log f_\mu \frac{dg_\mu}{g_\mu} )e_0.
}
We conclude that the Betti structure on $H$ is given upto a constant of integration by setting 
\eq{9.5}{h = \sum_\mu \int  \log f_\mu \frac{dg_\mu}{g_\mu}
}
in \eqref{8.2}. 

\begin{rmk} Note that the actual entries in \eqref{8.2} depend on the scaling of the $e_{i,\mu}$ which are given by algebraic de Rham classes. The actual values determined by the Feynman integrand will differ by an algebraic function of masses and momenta (eventually involving square roots) from the logs an dilogs in \eqref{8.2}. For an example of how this works, see section \ref{sect_phy}. 
\end{rmk}

\section{The Triangle Graph}\label{sect_tri}

The amplitude associated to the {\it triangle graph} with $1$ loop, $3$ vertices and $3$ internal edges, is of interest both physically and mathematically. Let $C, D \subset \P^2$ be rational curves. We assume they are reduced
but not necessarily irreducible. {\it Rational} in this context simply
means that the normalization of each irreducible component is
$\P^1$. Assume further that the intersection $C\cap D$ is
transverse. In particular, $C\cap D$ is a finite set of smooth points
in $C$ and in $D$. Let $C^0 = C- (C\cap D)$ (resp. $D^0 = D-(C\cap
D)$). The triangle graph yields a motive \eqref{8.6} which has the form
$H^2(\P^2-D,C^0)$ for suitable $C, D$. 
\begin{figure}
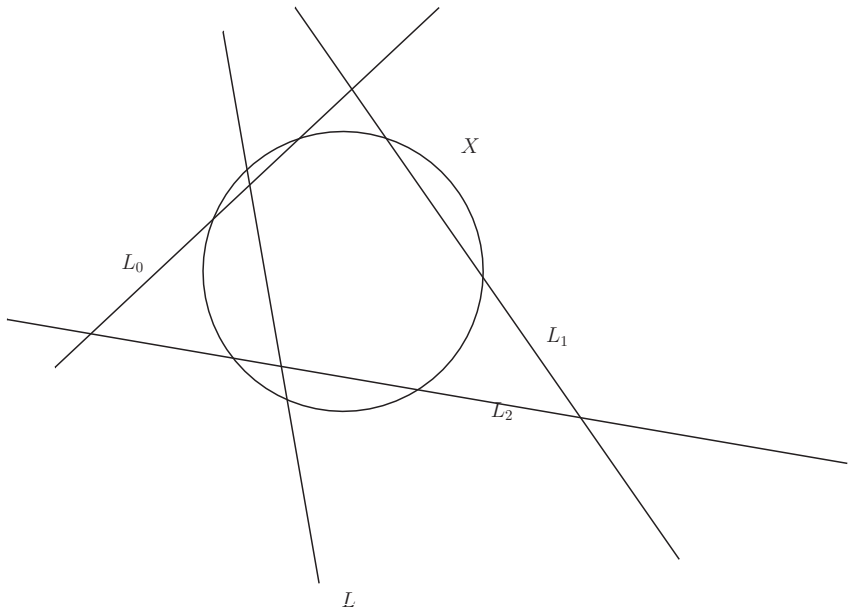

$$\triangleGeom$$
\caption{The geometry of the triangle graph. We indicate the three lines $L_0,L_1,L_2$, the line $L$, and the conic $X$.
The latter is given by $X:q_0^2A_1A_2+q_1^2A_2A_0+q_2^2A_0A_1-(m_0^2A_0+m_1^2A_1+m_2^2A_2)(A_0+A_1+A_2)=0$. 
They are all in general position. There are many degeneracies possible: for example, the conic would go through the three corners
$L_i=L_j$ in the massless case, or the conic can become tangential to one of those lines.} 
\end{figure}

We are particularly interested in the case when $C = L_0\cup L_1\cup L_2$ is the coordinate simplex (with homogeneous coordinates $A_i$ and $L_i:A_i=0$) and $D = L\cup X$ with $L:A_0+A_1+A_2=0$ and $X\subset \P^2$ a conic.  We write for simplicity
\eq{8.6}{H := H^2\Big(\P^2-(L\cup X),(L_0\cup L_1\cup L_2)-\big((L\cup X)\cap (L_0\cup L_1\cup L_2)\big),\Q\Big)
}
For the moment we assume that $X$ is a smooth conic in general position with respect to the other lines. 
\begin{prop}The Hodge structure on $H$ is mixed Tate, given by $W_0H \subset W_2H\subset W_4H$ with 
\eq{7.7}{gr^W_0H = \Q(0);\quad gr^W_2H = \Q(-1)^5;\quad gr^W_4H = \Q(-2).
}
\end{prop}
\begin{proof}Write $C = L_0\cup L_1\cup L_2$ and $D = L\cup X$. We have
\eq{7.8}{H^1(\P^2-D) \to H^1(C-C\cap D) \to H \to H^2(\P^2-D) \to 0
}
We have by Poincar\'e duality (formulated algebro-geometrically using cohomology with support, 
)
\eq{7.9}{H^2(\P^2-D,\Q) \cong H^3_D(\P^2,\Q) \cong H_1(D,\Q(-2))\cong \Q(-2). 
}
Note that topologically, $D$ is a union of two Riemann spheres $S^2$ meeting at two distinct points $p_1, p_2$. We get
\eq{7.10}{H_1(S^2,\Q(-2))^{\oplus 2} \to H_1(D,\Q(-2)) \to H_0(\{p_1,p_2\},\Q(-2)) \to H_0(S^2,\Q(-2))^{\oplus 2}
}
from which one deduces $H_1(D,\Q(-2)) \cong \Q(-2)$. We have again by duality a diagram
\eq{7.11}{\begin{CD}H^1(\P^2-D) @>>> H^1(C-C\cap D) \\
@VV \partial_D  V @VV \partial_C V \\
H_2(D,\Q(-2)) @>>> H_0(C\cap D, \Q(-1)) \\
@VV \cong V @VV \cong V\\
\Q(-1)^{\oplus 2} @>\inj >> \Q(-1)^{\oplus 9}
\end{CD}
}
The map $\partial_D$ is injective and has image the kernel of
$H_2(D,\Q(-2)) \to H_2(\P^2,\Q(-2))$ which is one dimensional. The
image of $\partial_C$ consists of all elements in $H_0(C\cap D,
\Q(-1))$ which have degree $0$ on each irreducible component of
$C$. The kernel of $\partial_C$ is $H^1(C,\Q) \cong \Q(0)$. The
proposition follows. (A detailed proof that $\dim gr^W_2H = 5$ is
given in remark \ref{rmk8.3} below.)\end{proof}
\begin{rmk}\label{rmk8.3} Let $\ell_i = L_i\cap L$ and $\{m_i,n_i\} = X\cap L_i,\ i=0,1,2$. We can identify $gr^W_2H$ with a subquotient of the Hodge structure $\Q(-1)^9$ with basis indexed by the $\ell_i, m_i, n_i$ as follows
\ml{}{gr^W_2H \cong \Q(-1)^5 \cong \\
\Big\{\sum_{i=0}^2 a_i\ell_i + b_im_i+c_in_i\ |\ a_i, b_i, c_i \in \Q(-1), a_i+b_i+c_i=0\Big\}\Big/ \Big\{\Q(-1)\cdot \sum_{i=0}^2 (2\ell_i-m_i-n_i)\Big\}. 
}
Alternatively, Consider zero cycles $z=\sum_0^2 a_i\ell_i + \sum b_im_i + \sum c_i n_i$ with $a_i, b_i, c_i \in \Q$. We impose the condition that for any one of the irreducible components $L, L_0, L_1, L_2, X$, the ``piece'' of $z$ supported on that component has degree $0$. This amounts to the linear conditions
\eq{8.13}{0 = \sum a_i = \sum (b_i+c_i) = a_0+b_0+c_0 = a_1+b_1+c_1 = a_2+b_2+c_2.
} 
The vector space $A$ of such cycles has dimension $5$ and is identified with $\Q(1)\otimes gr_2^W H$. 
\end{rmk}
\begin{rmk}\label{rmk8.4}
Let $M_i, i=1,2,3$ be the masses associated to the edges of the
triangle graph. There is physical interest in the situation when one
or more of the $M_i=0$. With reference to \eqref{2.16}, we see that
setting $M_i=0$ amounts to having the conic pass through the $i$-th
vertex of the triangle. The curves $C, D$ in the above discussion no
longer meet transversally, so we must blow up some of the vertices. Let $\pi: P \to \P^2$ be the blowup of $\nu=1, 2, 3$ of the three points $(1,0,0), (0,1,0), (0,0,1)$. Assume the
other parameters are generic and let $E_i$ be the exceptional
divisors. In our motive $H$ \eqref{8.6} we must replace $\P^2$ with $P$. 
The curve $L\cup X$ is replaced by the strict transform in $P$ of $L\cup X$, and the other rational curve 
becomes the total inverse image in $P$ of the triangle, a $(3+\nu)$-gon comprising the strict transforms in $P$ of the three lines $A_j=0$ and the exceptional divisors $E_i$. One checks that each blowup drops the rank of $gr_2^WH$ by one. Thus $gr_2^W H =
\Q(-1)^b$ with $b=5-\nu$, $\nu$ being the number of zero masses.

\end{rmk}
We will not compute the amplitude, as this has been done very nicely in \cite{DD}. Instead we will look more closely at qualitative results we can deduce about the motive. In particular, these will help to frame a (future) study of Landau poles for 1-loop graphs. 

\section{Appendix on duality for the triangle graph motive}
The following duality result, proposition \ref{prop_dual}, will not
be used in the sequel.

\begin{lem}Consider the diagram of varieties over $\C$ endowed with the complex topology. 
\eq{}{\begin{CD} \P^2-D @> j_D >> \P^2 \\
@A k_C AA @AA j_C A \\
\P^2-(D\cup C) @> k_D >> \P^2-C
\end{CD}
}
where the maps are the evident inclusions. We continue to assume $C\cap D$ is transverse. Let $A$ be a constant sheaf on $\P^2$. Then
\eq{7.2}{j_{D*}k_{C!}A_{\P^2-(D\cup C)} = j_{C!}k_{D*}A_{\P^2-(D\cup C)}. 
}
Here the lower $!$ and the lower $*$ are extension by zero and direct image extension viewed as acting on the derived category. (That is, we write e.g. $j_{D*}$ in place of $Rj_{D*}$.) 
\end{lem}
\begin{proof} There is a natural morphism of functors $j_{C!}k_{D*}\to J_{D*}k_{C!}$. Indeed, $j_{C!}$ is left adjoint to $j_C^! = j_C^*$, so it suffices to define a map $k_{D*} \to j_C^*j_{D*}k_{C!} = k_{C*}$, and we can take the identity. To show the map is an isomorphism, it suffices to look at the stalks at points of $C\cap D$. We can coordinatize a small complex neighborhood of such a point so locally $\P^2-(C\cap D)$ looks like $(U-\{0\})\times (U-\{0\})$ where $U \subset \C$ is the open unit disk about $0$. Locally, $D=\{0\}\times U$ and $C = U\times \{0\}$. The stalk at $(0,0)$ on both sides of \eqref{7.2} is $j_{C!}A\boxtimes j_{D*}A$ by Kunneth. 
\end{proof}

Recall we have a Verdier duality functor on the derived category of constructible sheaves on a reasonable topological space. For $\P^2_\C$ it takes the form (to simplify I work with sheaves of $\Q$-vector spaces) $\D F = Hom(F, \Q_{\P^2}(2)[4])$ where the Hom is in the derived category. Verdier duality yields an isomorphism
\eq{}{ R\Gamma(\P^2, \D F) \cong \text{Hom}_\Q(R\Gamma(\P^2, F),\Q)
}
We have $\D j_{C!}A = j_{C*}\D A$ and $D j_{C*} A= j_{C!}\D A$. Using the lemma we find for $A=\Q$ 
\ml{}{\text{Hom}_\Q(R\Gamma(\P^2, j_{D*}k_{C!}\Q_{\P^2-(D\cup C)}),\Q) \cong \text{Hom}_\Q(R\Gamma(\P^2, j_{C!}k_{D*}\Q_{\P^2-(D\cup C)}),\Q) \cong \\
R\Gamma(\P^2,j_{C*}k_{D!} \Q_{\P^2-(D\cup C)}(2)[4]). 
}
Taking $H^{-2}$ on both sides yields an isomorphism (duality)
\begin{prop}\label{prop_dual} With notation as above, we have
\eq{8.5}{\text{Hom}_\Q(H^2(\P^2-D,C-(C\cap D),\Q),\Q) \cong H^2(\P^2-C,D-(C\cap D),\Q(2)).
}
\end{prop}

\section{Landau Poles and Thresholds}\label{sectlim}

In this section we examine the phenomenon of Landau poles and normal and anomalous thresholds. 
\begin{ex}\label{ex10.1} Consider the case of the triangle graph. Changing notation slightly, we can rewrite \eqref{2.16} in this case
\eq{10.1a}{D(q,A) = |q_0|^2A_1A_2+|q_1|^2A_0A_2+|q_2|^2A_0A_1-(\sum_0^2 m_i^2A_i)(\sum_0^2 A_i). 
}
Fix an $i$ and suppose $|q_i|^2=(m_j\pm m_k)^2$. Then
\eq{10.2a}{D(q,A)|_{A_i=0} = -(A_j+A_k)(m_j^2A_j+m_k^2A_k)+(m_j\pm m_k)^2A_jA_k= -(m_jA_j\mp m_kA_k)^2.
}
Geometrically, our conic becomes tangent to the line $L_i: A_i=0$. If we look at the motive \eqref{8.6}, we see that this corresponds to a degenerate configuration, and we might reasonably expect the amplitude to become singular. In fact, a moment's reflection reveals a vast number of possible degenerations including situations where the conic itself degenerates to a union of $2$ lines through a point $p$, which may lie on one of the $L_i$, and situations where the conic passes through the point $L_i\cap L$ where $L:A_0+A_1+A_2=0$. One would like to better understand the behavior of the amplitude near these singularities.  
\end{ex}

Let us consider a generalization due to Cutkosky \cite{ItzZ} of the above example to more general graphs. If we write the amplitude in its usual (non-parametrized) form, we find an integral over $\R^q$ where $q$ is the loop number of the graph. The integrand has in the denominator a product of rank $4$ affine quadrics of the form (roughly) $(\vec{x}-\vec{q})^2+m^2$. These quadrics determine the polar locus of the integrand, and hence the motive whose realization will contain the amplitude as a period. In fact the motive can be taken to be the union of the projective closures of the quadrics. If we ignore what is happening at infinity and just consider the affine quadrics, we might expect degeneracies to occur for values of the parameter $\vec{q}$ where some subset of the affine quadrics do not meet transversally. In general, the locus of such $\vec{q}$ will form a divisor in the space of momenta, and our first job is to use elimination theory to find this divisor. 

To formulate things precisely, we fix a graph $\Gamma$. We write $H=H_1(\Gamma,\Q)$ and $E=\text{Edge}(\Gamma)$. As in \eqref{3.9c}, \eqref{3.10c} we have 
\eq{10.3a}{0 \to H\otimes \sA \to \sA^E \xrightarrow{\partial} \sA^{V,0}  \to 0. 
}
For $e\in E$  write $e^\vee: \sA^E \to \sA$ for the evident functional. Let me write (abusively) $e^{\vee,2}: \sA^E \to \R,\ \sum_\ve a_\ve \ve \mapsto a_e\bar a_e=a_{e,0}^2+a_{e,1}^2+a_{e,2}^2+a_{e,3}^2$. Given $q=\sum q_vv \in \sA^{V,0}$ (so $\sum_v q_v=0$)  consider the set 
\eq{10.4a}{H(q) :=\partial^{-1}(q) \subset \sA^E
}
Note that if $q\neq 0$ then technically $H(q)$ is not a vector space but a torsor under the vector space $H\otimes \sA$. In fact, $H(q)$ embodies the Feynman rule imposing relations for each vertex. In other words, if vertex $v$ lies on edges $e_1,\dotsc,e_p$ oriented to point toward $v$, and if $h \in H(q)$, then $\sum_i e_i^\vee(h) = q_v$. 

Let $S \subset E$ be a subset of edges, and let $T\subset S$ be such that the $\{e^\vee|H(q)\}_{e\in T}$ form a basis for the vector space spanned by $\{e^\vee|H(q)\}_{e\in S}$. To make life interesting, assume $T\neq S$. Let $m_e \in \R$ be a collection of masses, and 
consider the quadrics 
\eq{10.5a}{e^{\vee,2}-m_e^2: H(q) \to \R;\quad e \in S.
}
Let $X(S,q) \subset H(q)$ be defined by the vanishing of the quadrics \eqref{10.5a}. We want to identify the set of $q\in \sA^{V,0}$ such that $X(S,q)$ is not a smooth subvariety of $H(q)$ of codimension equal to $\# S$. 

For each $e\in T$ we get $4$ $\R$-linear functionals $e^{\vee,(0)},\dotsc,e^{\vee,(3)}: \sA^E \to \R$. If we consider the Jacobian matrix for the map 
\eq{10.6a}{(\ldots,e^{\vee,2}-m_e^2,\ldots)_{e\in S}: H(q) \to \R^{\# S}
}
It will have $\# S$ rows and $4\cdot\# T$ columns. For the row given by $e\in S$ we write $e^{\vee}|H(q) = \sum_{\ve\in T} c_{e,\ve} \ve^\vee + a_e(q)$ with $a_e(q)\in \sA$ and $c_{e,\ve} \in \R$. We have
\eq{}{e^{\vee,2}|H(q) = \sum_{i=0}^3(\sum_{\ve\in T} c_{e,\ve} \ve^{\vee,(i)}+a_e(q)^{(i)})^2
}
For $\tau \in T$ it follows that the entry of the jacobian matrix corresponding to $\partial/\partial \tau^{(i)}$ is
\eq{}{ 2c_{e,\tau}( \sum_{\ve\in T} c_{e,\ve} \ve^{\vee,(i)}+a_e(q)^{(i)})
}
It follows that the $4$ entries of the $e$-row corresponding to $\partial/\partial\tau$ yield exactly
\eq{10.7a}{2c_{e,\tau} e^{\vee}|H(q).
}
If we assume $S$ ordered so the elements of $T=\{\tau_1,\dotsc,\tau_p\}$ come first, the matrix \eqref{10.7a} evaluated at $h\in H(q)$ will look like
\eq{10.10a}{2\cdot\begin{pmatrix}\tau_1^\vee(h) & 0 & 0 & \hdots \\
0 & \tau_2^\vee(h) & 0 & \hdots \\
\vdots &\vdots &\vdots &\vdots \\
0 & \hdots & 0 &  \tau_p^\vee(h) \\
\vdots &\vdots &\vdots &\vdots\end{pmatrix}
}
The point $h$ will be singular in $X(S,q)$ if first of all the quadrics \eqref{10.5a} vanish at $h$ (i.e. $h\in X(S,q)$) and secondly there exists a real non-zero row vector $\vec{b}=(b_1,\dotsc,b_{\#S})$ of length $\#S$ which dies under right multiplication by \eqref{10.10a}. Since the $e^\vee|H(q)$ are affine linear combinations of the $\tau_i^\vee|H(q)$ we can use such a vector, which we treat as a vector with unknown entries $b_i$, to write $p$ affine linear equations
\eq{}{\sum_{i=1}^p\alpha_{ij}(\vec{b})\tau_i^\vee(h) = \beta_j(\vec{b},q);\quad j=1,\dotsc,p
}
We then solve these equations:
\eq{}{\tau_i^\vee(h) = \gamma_i(\vec{b},q)
}
and substitute into the quadrics \eqref{10.5a} (again using that $e^\vee|H(q)$ are affine linear combinations of the $\tau_i^\vee|H(q)$).  Note that the $\gamma_i$ are homogeneous of degree $0$ in the $b_j$. The quadrics yield $\#S$ equations $F_i(\vec{b},q)=0$ which are homogeneous of degree $0$ in the $b_i$. Write $\A^{V,0}$ for the affine space associated to $\sA^{V,0}\cong (\R^{V,0})^4$. We can view the $b_i$ as homogeneous coordinates on $\P^{\#S-1}$. In this way we get $\#S$ equations in $\P^{\#S-1}\times \A^{V,0}$. Projecting down to $\A^{V,0}$ amounts to eliminating the variables $b_i$. The image is a closed subvariety $Z\subset \A^{V,0}$ with the property that $q\in Z \Leftrightarrow \text{the intersection of the quadrics in \eqref{10.5a} is not transverse}.$ Note that in general we expect $Z$ is a hypersurface in $\A^{V,0}$ though of course degeneracies can occur. This $Z$ is our divisor. 
\begin{ex}\label{ex10.2} Suppose elements of $S$ form a cut, i.e. that $\Gamma - \bigcup_{e\in S} e$ is disconnected but that $S$ is minimal in the sense that no proper subset of $S$ disconnects $\Gamma$. (In removing edges, we do not remove vertices, so one of the connected components may be an isolated vertex.) It is easy to see in this case that $S-T=\{e\}$ is a single edge, so $S=\{\tau_1,\dotsc,\tau_p,e\}$ and $\#S=p+1$. For a suitable edge orientation we get
\eq{}{e^\vee = \sum_{i=1}^p \tau_i^\vee + a(q)
} 
where $a(q)$ is some fixed linear combination (depending on $\Gamma,\ S,\ T$) of the $q_v$. (Recall $q=\sum q_vv$.) The matrix \eqref{10.10a} in this case is
\eq{10.14a}{2\cdot\begin{pmatrix}\tau_1^\vee & 0 & 0 & \hdots \\
0 & \tau_2^\vee & 0 & \hdots \\
\vdots & \vdots & \vdots &\hdots \\
0 & \hdots & 0 & \tau_p^\vee \\
\sum_{i=1}^p \tau_i^\vee + a(q) & \sum_{i=1}^p \tau_i^\vee + a(q) &  \hdots & \sum_{i=1}^p \tau_i^\vee + a(q)\end{pmatrix}.
}
The linear equations and their solutions become
\ga{10.15a}{(b_i+b_{p+1})\tau_i^\vee +b_{p+1}(\sum_{j\neq i}\tau_j^\vee + a(q))=0;\quad i=1,\dotsc,p \\
\tau_i^\vee = a(q)D_i(b)/D(b) \label{10.16a}
}
It is easy to see the determinant $D(b)$ in the denominator does not identically vanish because the term $b_1b_2\cdots b_p$ cannot cancel.

The quadrics in this case become after substitution ($|a|^2=a\bar a,\ a\in \sA$.)
\ga{10.17a}{|a(q)|^2D_i(b)^2/D(b)^2=m_i^2;\quad 1\le i\le p \\
|a(q)|^2(1+\sum_{i=1}^pD_i(b)/D(b))^2 = m_{p+1}^2 \label{10.18a}
}
Combining these, we deduce finally
\eq{10.19a}{|a(q)| = \sum_{i=1}^{p+1}\mu(i)m_i;\quad \mu(i) = \pm 1.
}
Note \eqref{10.19a} is necessary and sufficient for the intersection of the Feynman quadrics on $H(q)$ indexed by $S$ to be non-transverse. Indeed, if \eqref{10.19a} holds, we can solve for the $\tau_i^\vee$ as multiples of $a(q)$ using \eqref{10.16a}. The resulting point will lie on $X(S,q)$. The matrix \eqref{10.14a} can then be treated as a matrix of scalars (more precisely, all entries lie on the same line). It has $p+1$ rows and $p$ columns, so there is necessarily a non-trivial solution $\vec{b}$ and the point on $X(S,q)$ is not a point of transverse intersection.

The values of $q$ where \eqref{10.19a} hold are called {\it normal thresholds}. We have seen \eqref{10.2a} that in the case of the triangle graph, normal thresholds correspond to values of external momenta where the polar conic \eqref{10.1a} becomes tangent to one of the $L_i:A_i=0$. 
\end{ex}

\section{Thresholds for the triangle graph}\label{sect_lim2}
In this section, we outline the use of limiting mixed Hodge structures (cf. section \ref{sectHS}) to study thresholds. We take the very basic case of the triangle graph with zero masses. The second Symanzik polynomial has the form
\eq{}{Q=|q_0|^2A_1A_2+|q_1|^2A_0A_2+|q_2|^2A_0A_1.
}
We will eventually want to assume $q_i=q_i(t)$ for $t$ in a small disk about $0$ and that two of the $|q_i(t)|^2$ tend to $0$ as $t\to 0$. (The notation here is misleading. $q_i$ is a $4$-vector and $|q_i|^2 = \langle q_i,q_i\rangle$ for some non-degenerate quadratic form on $\C^4$. In particular, $|q_i|^2$ is analytic in $q_i$.) Our objective will be to show in this case that the logarithm of monodromy $N$ \eqref{2.10b} satisfies $N^2 \neq 0$ and that as a consequence the leading term for the expansion of the amplitude as $t\to 0$ is a non-zero multiple of $(\log t)^2$.

The differential form we need to integrate is
\eq{11.2d}{\eta(q):= \frac{\Omega_2}{(A_0+A_1+A_2)(|q_0|^2A_1A_2+|q_1|^2A_0A_2+|q_2|^2A_0A_1)}
} 
Let $\pi: P \to \P^2$ be the blowup of the three vertices $A_i=A_j=0$. We take $q$ general so the singularities of the polar locus of $\eta(q)$ do not fall at the vertices. Let $Z(q) \subset P$ be the strict transform of this polar locus. Let $E_0, E_1, E_2\subset P$ be the exceptional divisors, so $E_i$ lies over $A_j=A_k=0$, and let $F_i\subset P$ be the strict transform of the locus $\{A_i=0\}$.  The union
\eq{}{\Sigma := \pi^*\Delta=E_0\cup E_1\cup E_2\cup F_0\cup F_1\cup F_2 \subset P
}
forms a hexagon. Note that $Z(q) = L'\cup Y(q)$ where $L'$ is the strict transform of the line $L:A_0+A_1+A_2=0$ in $\P^2$ and $Y(q)$ is the strict transform of the conic. $L'$ meets each $F_i$ in a single point, and $Y(q)$ meets each of the $E_i$ in a single point. Let $\Sigma^0 := \Sigma - \Sigma\cap Z(q)$. Then $\Sigma^0$ is a hexagon of affine lines $E_i^0, F_j^0$, so $H^1(\Sigma^0)=\Q(0)$. The motive we need to study is 
\eq{}{H:= H^2(P-Z(q),\Sigma^0).
}

We have seen in Remark \ref{rmk8.4} that $gr^WH = \Q(0)\oplus \Q(-1)^2 \oplus Q(-2)$. The next step is to construct the Kummer motives $W_2H$ and $H/W_0H$ (see Example \ref{kumex}).  Let $S=\sum n_is_i$ be a $0$-cycle (formal linear combination of smooth points) on $\Sigma^0$. We define a Kummer extension $K_S$ by pullback as follows \minCDarrowwidth.1cm
\eq{}{\begin{CD} 0 @>>> H^1(\Sigma^0) @>>> H^1(\Sigma^0-\{s_i\}) @>>> \bigoplus_i \Q(-1) @>>> 0 \\
@. @| @AAA @AA s A \\
0 @>>> \Q(0) @>>> K_S @>>>  \Q(-1) @>>> 0.
\end{CD}
}
Recall \eqref{2.10d} that Kummer extensions $\leftrightarrow \C^\times$. Let $[S] \in \C^\times$ correspond to $K_S$ as above. It is an easy exercise to check that the mapping
\eq{}{\{\text{$0$-cycles on $\Sigma^0$}\} \to \C^\times;\quad S \mapsto [S]
}
is a homomorphism of groups. To compute this map, we note that the $E_i, F_i$ are projective lines with natural projective coordinates $a_j, a_k$. We have
\eq{}{F_i^0 = F_i-\{-1\};\quad E_i^0 = E_i - \{|q_j|^2a_k+|q_k|^2a_j=0\}.
}
Suppose $S=s$ is a single point. If $s\in E_i^0$ (resp. $s\in F_i^0$) we choose $f$ a regular function on $E_i^0$ (resp. $F_i^0$) with a simple zero at $s$ and no other zeroes. We then orient our hexagon by ordering the edges $E_0, F_1, E_2, F_0, E_1, F_2$. Let $j, k$ be such that $F_j, E_i, F_k$ (resp. $E_j, F_i, E_k$) is part of the ordered string of edges. Define $[s] = \frac{f(E_i\cap F_k)}{f(E_i\cap F_j)}$ (resp.  $[s] = \frac{f(F_i\cap E_k)}{f(F_i\cap E_j)}$.) 

We will be interested in the case $S=H_i'\cdot \Sigma^0$ where $H_i'$ is the strict transform of the line $H_i: A_j-A_k=0$ in $\P^2$. Thus $S=\{1\in E_i^0\} + \{1\in F_i^0\}$. On $F_i$ we take $f=\frac{a_j-a_k}{a_j+a_k}$, so $[ 1\in F_i^0] = -1$. On $E_i$, let $f=\frac{a_j-a_k}{|q_j|^2a_k+|q_k|^2a_j}$, so $[1\in E_i^0] = \frac{-|q_j|^2}{|q_k|^2}$. Taken together, we see 
\eq{11.8c}{[H_i'\cdot \Sigma^0] = \frac{|q_j|^2}{|q_k|^2}. 
}
\begin{lem}\label{lem11.1} $W_2H$ is an extension of $\Q(-1)^2$ by $\Q(0)$ corresponding to extension classes 
\eq{11.9a}{\frac{|q_j|^2}{|q_k|^2}, \frac{|q_i|^2}{|q_k|^2} \in \C^\times.
}
Here $i, j, k$ are all distinct. 
\end{lem}
\begin{proof}We have a diagram \minCDarrowwidth.1cm
\eq{11.9c}{\begin{CD}0 @>>> H^1(\Sigma^0) @>>> H^2(P-Z(q),\Sigma^0) @>>> H^2(P-Z(q)) @>>> 0 \\
@. @| @A a AA @A d AA \\
0 @>>> H^1(\Sigma^0) @>>> H^2(P-Z(q)\cap \Sigma^0,\Sigma^0) @>>> H^2(P-Z(q)\cap \Sigma^0 ) @>>> 0 \\
@. @. @A c AA @A b AA \\
@. @. \Q(-1)^2 @= \Q(-1)^2 \\
@. @. @AAA @AAA \\
@. @. 0 @. 0.
\end{CD}
}
$W_2H$ is the image of the map $a$. Also
\eq{}{H^2(P-Z(q)\cap \Sigma^0 )\cong H^2(P) \cong \Q(-1)^4
}
generated by the $4$ divisor classes $[L'], [E_0], [E_1], [E_2]$. The map $b$ has image generated by the two divisor classes $[L'], [Y(q)] = 2[L']-[E_0]-[E_1]-[E_2]$. Note that $b$ lifts to a map $c$ as indicated because the divisors $L', Y(q)$ do not meet $\Sigma^0$. It follows that the image of the map $d$ is generated by the divisor classes $H_i'$ where the $H_i'$ are as above. The lemma follows from \eqref{11.9a}.
\end{proof}

It is convenient to work with $H_i'-H_j'$. If, e.g., we restrict the extension given by the top line of \eqref{11.9c} to $\Q(-1)(H_0'-H_2') \subset H^2(P-Z(q))$, the resulting Kummer extension by the lemma is $\frac{|q_2|^2|q_0|^2}{|q_1|^4}$. Similarly, the extension class after restriction to $\Q(-1)(H_0'-H_1')$ is $\frac{|q_2|^4}{|q_0|^2|q_1|^2}$. (We are using here the orientation of the hexagon as fixed above.)

\begin{lem}$H/W_0$ is an extension of $\Q(-2)$ by $\Q(-1)^2$ corresponding to extension classes given by formulas \eqref{11.15a}, \eqref{11.16a} below. 
\end{lem}
\begin{proof}We can identify 
\eq{10.29}{H/W_0 \cong H^2(P-Z(q)).  
}
Here $Z(q)$ is isomorphic to the union of the conic $X(q): |q_0|^2A_1A_2+|q_1|^2A_0A_2+|q_2|^2A_0A_1=0$ and the projective line $A_0+A_1+A_2=0$. We take the coefficients $|q_i|^2$ to be general, so these two plane curves meet in $2$ distinct points:
\ml{}{p_{\pm} : A_0=-A_2-A_1; \\A_1 = \frac{|q_0|^2-|q_1|^2-|q_2|^2 \pm \sqrt{|q_0|^4+|q_1|^4+|q_2|^4-2|q_0|^2|q_1|^2-2|q_0|^2|q_2|^2-2|q_1|^2|q_2|^2}}{2|q_2|^2}A_2.
}

We will also need (straightforward check) that the function $f_{01}:= 1 + \frac{|q_2|^2A_1}{|q_1|^2A_2}$ on $X(q)$ has divisor $(f_{01}) = (1,0,0) - (0,1,0)$. Similarly, $f_{02} = A_2f_{01}/A_1$ has divisor $(f_{02})=(1,0,0)-(0,0,1)$. 

Using the techniques of section \ref{sect_tri} and \eqref{10.29} we can identify $H/W_0$ with the extension
\ml{11.14c}{0 \to \Big(\Q(-1)v_0+\Q(-1)v_1+\Q(-1)v_2\Big)\Big/\Q(-1)\Big(v_0+v_1+v_2\Big) \to \\
 H_1\Big(Z(q)-\{v_0,v_1,v_2\},\Q(-2)\Big) \to H_1(Z(q),\Q(-2)) \to 0.
}
More directly, If we identify 
\eq{}{\text{Image}(H^2(P) \to H^2(P-Z(q))) \cong \Q(-1)(H_0'- H_2')\oplus \Q(-1)(H_1'- H_2')
}
we can deduce from \eqref{10.29} an exact sequence
\eq{11.16c}{0 \to \Q(-1)(H_0'- H_2')\oplus \Q(-1)(H_1'- H_2') \to H/W_0 \to H_1(Z(q),\Q(-2)) \to 0
}
We have $[H_i'-H_j'] = [E_j-E_i]$, and the identification of \eqref{11.16c} with \eqref{11.14c} sends $[H_i'-H_j'] \mapsto v_j-v_i = [E_j-E_i]\cdot Z(q)$. 

Twisting and dualizing \eqref{11.14c}, we get the extension
\eq{10.32}{0 \to H^1(Z(q),\Q) \to H^1(Z(q)-\{v_0,v_1,v_2\},\Q) \to \Big(\Q(-1)v_0+\Q(-1)v_1+\Q(-1)v_2\Big)^{\deg 0} \to 0.
} 
The class of the extension obtained by restricting on the right to $\Q(-1)(v_i-v_j)$ is calculated by the ratio $f_{ij}(p_+)/f_{ij}(p_-) \in \C^\times$. We have, e.g.
\ml{11.15a}{f_{01}(p_+)/f_{01}(p_-)= \\
\frac{\Big(|q_0|^2+|q_1|^2-|q_2|^2+\sqrt{|q_0|^4+|q_1|^4+|q_2|^4-2|q_0|^2|q_1|^2-2|q_0|^2|q_2|^2-2|q_1|^2|q_2|^2}\Big)^2}{4|q_0|^2|q_1|^2}.
}
\ml{11.16a}{f_{02}(p_+)/f_{02}(p_-)= \\
\frac{|q_0|^2+|q_2|^2-|q_1|^2-\sqrt{|q_0|^4+|q_2|^4+|q_1|^4-2|q_0|^2|q_2|^2-2|q_0|^2|q_1|^2-2|q_2|^2|q_1|^2}}{|q_0|^2+|q_2|^2-|q_1|^2+\sqrt{|q_0|^4+|q_2|^4+|q_1|^4-2|q_0|^2|q_2|^2-2|q_0|^2|q_1|^2-2|q_2|^2|q_1|^2}} \times \\
 f_{01}(p_+)/f_{01}(p_-)      \\
}
\end{proof}

Suppose now that the $|q_i|^2 = |q_i(t)|^2$ are analytic functions of a parameter $t$ with $|t|<\ve$. After scaling we may suppose $\lim_{t\to 0} |q_0(t)|^2=1$. We will suppose further that $\text{ord}_0(|q_i|^2)>0$ for $i=1,2$. 
Replacing $t$ by a power if necessary, we can arrange that the monodromy $\sigma$ as $t$ winds around $0$ acts trivially on $gr^WH$. We want to compute $N^2=(\log\sigma)^2$. For the family of Kummer extensions $E_{x(t)}$ as in example \ref{kumex} with $gr^WE = \Q(1)\oplus \Q(0)$, one sees easily from \eqref{2.8b} that $N$, viewed as a map $\Q(0) \to \Q(1)(-1) = \Q(0)$ is multiplication by $\text{ord}_0(x(t))$. Similarly, $N^2: H \to H(-2)$ factors as $N^2 : \Q(-2) = H/W_2 \to W_0(-2) = \Q(-2)$. This in turn can be factored
\eq{}{\Q(-2) \xrightarrow{N_{H/W_0}} \Q(-2)(H_0'-H_2')\oplus \Q(-2)(H_0'-H_1') \xrightarrow{N_{W_2}} \Q(-2).
}
Set 
$$b = \text{ord}_0\Big(|q_0|^2+|q_2|^2-|q_1|^2-\sqrt{|q_0|^4+|q_2|^4+|q_1|^4-2|q_0|^2|q_2|^2-2|q_0|^2|q_1|^2-2|q_2|^2|q_1|^2}\Big) > 0.
$$
$$c = \text{ord}_0 f_{01}(p_+)/f_{01}(p_-).
$$

Formula \eqref{11.16a} and lemma \ref{lem11.1} (see also the discussion below that lemma) one sees that 
$N^2: \Q(-2) \to \Q(-2)$ is multiplication by 
\eq{11.21a}{ \rho := (4\cdot\text{ord}_0(|q_2|^2)-2\cdot\text{ord}_0(|q_1|^2))c + (-4\cdot\text{ord}_0(|q_1|^2)+2\cdot\text{ord}_0(|q_2|^2))(b+c). 
}
If for example we take $\text{ord}_0(|q_1|^2)=\text{ord}_0(|q_2|^2)>0$, we get 
\eq{}{N^2 = \text{multiplication by} -2b\cdot\text{ord}_0(|q_2|^2)\neq 0. 
}

In order to explicit the limiting behavior of the amplitude, we consider \eqref{2.13b}, which in the current 
setup ($\eta(q)$ as in \eqref{11.2d}) looks like
\eq{}{\lim_{t\to 0} \exp\left(-N\frac{\log t}{2\pi i}\right)\begin{pmatrix}\int_{\gamma_0}\eta(q) \\ \int_{\gamma_{-2,1}}\eta(q) \\ \int_{\gamma_{-2,2}}\eta(q) \\  \int_{\gamma_{-4}}\eta(q)\end{pmatrix} = \begin{pmatrix} a_0 \\ a_{1,1} \\ a_{1,2} \\ a_2 \end{pmatrix}
}
Here, the $\gamma_j$ for $j\le i$ form a basis for the homology $H^\vee_{\Q}$. Our limiting approximation for $\int_{\gamma_0}\eta(q)$ is therefore the top entry in the column vector
\eq{}{ \exp\left(+N\frac{\log t}{2\pi i}\right)\begin{pmatrix} a_0 \\ a_{1,1} \\ a_{1,2} \\ a_2 \end{pmatrix}.
}
Since $N^2$ has all entries $0$ except for $\rho$ \eqref{11.21a} in the upper right corner, we conclude
\eq{11.25a}{\int_{\gamma_0}\eta(q) \sim \frac{\rho}{2}\Big(\lim_{t\to 0} \int_{\gamma_{-4}}\eta(q(t))\Big)(\log t/2\pi i)^2 + B\log t /2\pi i+ C
}
for suitable constants $B, C$. 

It remains, finally to compute the limit in \eqref{11.25a}. The cycle $\gamma_{-4}$ is a generator of the image $(W_4H)^\vee \inj H^\vee$. By adjunction, we can compute the integral by a suitable residue computation on $\eta(q)$. In affine coordinates $a_i = A_i/A_0$ we find
\eq{}{\pm \eta(q) = \frac{da_1\wedge da_2}{(1+a_1+a_2)(|q_0|^2a_1a_2+|q_1|^2a_2+|q_2|^2a_1)}
}
Let $b_2$ be the coordinate $a_2$ restricted to the line $a_1+a_2+1=0$. The residue yields
\eq{11.27a}{\frac{db_2}{-|q_0|^2b_2^2+(|q_1|^2-|q_0|^2-|q_2|^2)b_2-|q_2|^2}
}
For a suitable choice of $\gamma_{-4}$, $\int_{\gamma_{-4}}\eta(q(t))$ will be the difference of the two residues of \eqref{11.27a}. Since the sum of the two residues is zero, it will suffice to show that an individual residue does not tend to $0$.  With our assumptions that $|q_0|^2 \to 1, |q_i|^2\to 0, i=1,2$, this is straightforward. We have proved
\begin{thm} Consider the triangle graph with zero masses and momenta $q_i(t), i=0,1,2$. We treat the momenta as complex $4$-vectors, so $|q_i(t)|^2 = \sum_{j=1}^4 q_i^{(j)}(t)^2$ is analytic in a complex parameter $t$ for $|t| \to 0$. Assume $|q_0(0)|^2=1$ and $|q_i(0)|^2=0,\ i=1,2$. Assume further that the $\text{ord}_0(|q_i|^2)$ are such that $\rho$ in \eqref{11.21a} is non-zero. (E.g. $|q_1|^2$ and $|q_2^2$ vanish to equal order at $t=0$). Then if we take $\gamma_0$ to be the chain $\{(x,y,z) \in \P^2(\R)\ |\ x,y,z\ge 0\}$ in $\P^2$ then 
\eq{}{\int_{\gamma_0}\eta(q) \sim A(\log t/2\pi i)^2 + B\log t/2\pi i + C
}
for suitable constants $A, B, C$ with $A\neq 0$. 
\end{thm}

\section{Physics}\label{sect_phy}

Let us now try to understand the above considerations from a physicists viewpoint.
Setting an edge variable to zero turns the triangle graphs into three reduced diagrams
$$\rtrc,\rtrb,\rtra.$$

Each of them is a function of a single invariant $q_1^2,q_2^2$ or $q_3^2$.
The computation of these reduced diagrams is straightforward and delivers
(in the equal mass case, otherwise the Kallen function replaces the square root) a result of the form
\begin{equation}
 \sqrt{1-\frac{4m^2}{q^2}}\ln\frac{\sqrt{1-\frac{4m^2}{q^2}}-1}{\sqrt{1-\frac{4m^2}{q^2}}+1},
\label{residue}\end{equation}
which has, as a function of $q^2$, a branchcut from $[4m^2,+\infty[$ and a 
variation there $\sim \sqrt{1-\frac{4m^2}{q^2}}$.

This gives us the equivalent of the functions $f_i$ above. Notice, 
however, that the expected $\log$ is multiplied by an algebraic function. 
The problem comes in the normalization of $e_{-1}$, or in other words the choice 
of differential form to represent a class in de Rham cohomology. 
Essentially, the Feynman integrand in this case has the form \cite{BlKr1} 
(in the equal mass case, and for the example that edge 2 shrinks)
\eq{}{
\omega=\frac{\ln \frac{m^2(A_0^2+A_1^2)+(2m^2-q_2^2)A_0A_1}{m^2(A_0^2+A_1^2)+(2m^2-\mu^2)A_0A_1}\Omega_1}{(A_0+A_1)^2}
}
where we renormalize at a renormalization point $q_2^2=\mu^2$ by a simple subtraction in this one-loop example.

An easy partial integration (the boundary terms do not contribute as we have a renormalized integrand)
 determines the result Eq.(\ref{residue}),
with the square root determined from the two solutions of the quadric.

The issue for the functions $g_j$ in \eqref{8.2} is more subtle. The idea is that the span of columns in \eqref{8.2} starting from the right hand column is supposed to be invariant under monodromy. In particular, the monodromy of the $h$ is supposed to be a linear combination of the $2\pi i\log g_j$. To mimic this in physics we may use Cutkosky cuts. Concretely,
we look at
$$\triangleab,\triangleca,\trianglebc.$$
They correspond to integrals
$$  \int d^4k \Theta(k_0+q_{i,0})\Theta(k_0+k_{j,0})\delta((k+q_i)^2-m^2)\delta((k+q_j)^2-m^2)
\frac{1}{k^2-m^2}$$
which readily integrate to 
$$ \int_a^b \frac{du}{cu+d}$$ for suitable $a,b,c,d$ depending on masses and external momenta
(these $a,b,c,d$ are in the literature, in \cite{KreimerOld} for example).

Finally, the completely cut $\triangleabc$ leaves no integral to be done, but gives a known rational function of
the $q_i^2,m_j^2$.

Hence, from a physicists viewpoint, the above structure looks like
{
$$
\left(\begin{array}{ccccc}
1 & 0 & 0 & 0 & 0\\
\rtrc & \rtrab  & 0 & 0 & 0\\
\rtrb & 0 & \rtrca  & 0 & 0\\
\rtra & 0 & 0 & \rtrbc  & 0\\
{\triangle} & \triangleab & \triangleca & \trianglebc & \triangleabc
\end{array}\right)
\!=\!(C_1,C_2,C_3,C_4,C_5)$$
nicely expressed in terms of reduced diagrams, Cutkosky cuts, and a traingle with all edges cut, which 
delivers a momentum and mass dependent constant as the right lowermost entry, corresponding to the entry 
$(2\pi i)^2$
in the classical dilog case. This justifies recent practice in physics to put more internal edges on the mass-shell than
prescribed by Cutkosky. 

A challenge for the future is to identify the correct differential equations and the connection with Griffith's 
transversality, so that it makes sense to discuss constructs like
$$\mathrm{Var}\left(\Im \triangle-\left[\Re \trianglebc \cdot \Im \rtra
\right] +\cdots\right)=0,$$
as functions of complex external momenta.

We believe it is basically the presence of such invariant functions in the complex domain which 
allows to analytically continue Feynman diagrams in a way which will make the analytic requirements 
on Green functions more transparent once the Hodge structures of terms in the perturbative expansion are under control.

\newpage \bibliographystyle{plain} \renewcommand\refname{References}

\begin{thebibliography}{99}
\bibitem{BlKr1}
S.\ Bloch, D.\ Kreimer, 
{\em Mixed Hodge Structures and Renormalization in Physics,}
Commun.\ Num.\ Theor.\ Phys.\ 2:637-718,2008. 
\bibitem{D}P.\ Deligne, Equations Diff\'erentielles \`a Points singuliers R\'eguliers, SLN {\bf 163}, Springer (1970).
\bibitem{D2}P.\ Deligne, Le symbole moder\'e, Publications math\'ematiques de l'I.H.E.S, tome 73(1991), p. 147-181.
\bibitem{DD}
 A.~I.\ Davydychev and R.\ Delbourgo,
  {\em A geometrical angle on Feynman integrals,}
  J.\ Math.\ Phys.\  {\bf 39} (1998) 4299
  [arXiv:hep-th/9709216].
\bibitem{ItzZ} J.-C.\ Itzykson, J.-B.\ Zuber, Quantum Field Theory. Mc-Graw-Hill, 1980.
\bibitem{ELOP}R.J.\ Eden, P.V.\ Landshoff, J.C.\ Polkinghorne, D.I.\ Olive, The analytic S-matrix, Cambridge Univ.\ Press 1966.
\bibitem{KreimerOld} D.\ Kreimer,
{\em One Loop Integrals Revisited. 2. The Three Point Functions,}
  Int.\ J.\ Mod.\ Phys.\  A {\bf 8} (1993) 1797.
\bibitem{Nickel}
B.G.\ Nickel,
{\em Evaluation of Simple Feynman Graphs,}
J.Math.Phys.19:542-548,1978. 
\bibitem{T} 
J.-P.\ Tignol, 
{\em Pfaffiens et déterminant de E. H. Moore},
Bull. Belg. Math. Soc. Simon Stevin Volume 6, Number 4 (1999), 537-539. 
\bibitem{P} E.\ Patterson, Thesis, Department of Mathematics, University of Chicago (2009). 
\bibitem{Smirnov} V.\ Smirnov, Evaluating Feynman Integrals, Springer Tracts in Modern Physics {\bf 211}. Berlin, Springer. ix, 247p. (2004).
\bibitem{Tod} I.\ Todorov, Analytic Properties of Feynman Diagrams in Quantum Field Theory, Monographs in Natural Philosophy {\bf 38}, Pergamon Press (1971).
\end{thebibliography}

\end{document}